\documentclass[journal]{IEEEtran}

\usepackage{cite}
\usepackage{amsmath,amssymb,amsfonts}
\usepackage{algorithmic}
\usepackage{graphicx}
\usepackage{textcomp}
\usepackage{xcolor}	
\usepackage[caption=false,font=footnotesize]{subfig}
\usepackage{stfloats}
\usepackage{epstopdf}
\usepackage{acronym}
\usepackage[bookmarksopen=true]{hyperref}
\usepackage{mathrsfs}
\usepackage{float}
\usepackage{placeins}
\usepackage{epstopdf}
\usepackage{amsthm}
\usepackage{commath}
\usepackage{booktabs}
\usepackage{flushend} 
\usepackage{lettrine}

\def\BibTeX{{\rm B\kern-.05em{\sc i\kern-.025em b}\kern-.08em
		T\kern-.1667em\lower.7ex\hbox{E}\kern-.125emX}}

\hypersetup{
	colorlinks=true,    
	allcolors=blue,
	urlcolor=black     
}

\newtheoremstyle{customdef} 
{}                          
{}                          
{\itshape}                         
{}                         
{\bfseries}                 
{.}                         
{ }                         
{}  

\theoremstyle{customdef}
\newtheorem{definition}{Definition}   

\newtheoremstyle{customproposition} 
{}                          
{}                          
{\itshape}                         
{}                         
{\bfseries}                 
{.}                         
{ }                         
{}  

\theoremstyle{customproposition}
\newtheorem{proposition}{Proposition}   

\newtheoremstyle{customremark} 
{}                          
{}                          
{}                         
{}                         
{\bfseries}                 
{.}                         
{ }                         
{}                          

\theoremstyle{customremark}
\newtheorem{remark}{Remark}  

\newtheoremstyle{customexample} 
{}                          
{}                          
{}                         
{}                         
{\bfseries}                 
{.}                         
{ }                         
{}                          

\theoremstyle{customexample}
\newtheorem{example}{Example}

\newcommand{\trans}{^{\mathrm{T}}}
\newcommand{\herm}{^{\mathrm{H}}}
\newcommand{\vecsym}[1]{\boldsymbol{\rm{#1}}}

\newcommand{\raisedZero}{^{(0)}}
\newcommand{\raisedOne}{^{(1)}}

\DeclareMathOperator*{\argmax}{arg\,max}

\def\e{\mathrm e}

\def\R{\mathbb{R}}
\def\C{\mathbb{C}}
\def\N{\mathbb{N}}

\def\binaryField{\mathbb{F}_2}

\def\E{\mathbb{E}}
\def\vector{\vecsym{v}}

\def\complexnormal{\mathcal{CN}}

\def\PAPR{\text{PAPR}}
\def\MSE{\text{MSE}}
\def\Arg{\mathrm{Arg}}

\def\bMessage{\vecsym{b}}
\def\bit{b}
\def\estimatedBit{\hat{b}}
\def\zero{\alpha}

\def\zeroIndex{k}
\def\numZeros{K}
\def\zeroMag{\rho}
\def\zeroPhase{\psi}
\def\radius{R}
\def\baseAngle{\theta_\numZeros}

\def\bZero{\boldsymbol{\zero}}
\newcommand{\crPair}[1]{\mathscr{Z}_{#1}}

\def\polyCodebook{\mathscr{C}^\numZeros}

\def\X{X}
\def\H{H}
\def\W{W}
\def\Y{Y}
\def\Ymod{\tilde{\Y}}
\def\Xhat{\hat{\X}}

\def\txSeq{\vecsym{x}}
\def\rxSeq{\vecsym{y}}

\def\noiseSeq{\vecsym{w}}

\def\channelZero{\beta}
\def\noiseZero{\gamma}

\def\A{A}
\def\Ahuffman{\A_\mathrm{H}}
\def\Ajutted{\A_{\mathrm{J}}}
\def\Xhuffman{\X_\mathrm{H}}
\def\Xjutted{\X_{\mathrm{J}}}
\def\etaH{\eta_{_\mathrm{H}}}
\def\etaJ{\eta_{_\mathrm{J}}}
\def\autoIndex{\ell}

\def\asymFactor{\zeta}

\def\pseudoLLR{\text{PLLR}}

\def\dataSymbol{d}

\def\numSubcarriers{S}
\def\subcarrierIndex{\ell}
\def\numOFDMsymbols{M}
\def\ofdmSymbolIndex{m}

\def\idftSize{N}
\def\idftIndex{n}
\def\idftIndexHat{\hat{\idftIndex}}

\def\cpSize{N_\mathrm{cp}}

\def\txSignal{s}
\def\rxSignal{r}

\def\rxSignalMod{\tilde{r}}

\def\sampleRate{f_\mathrm{s}}

\def\symbolDuration{T_\mathrm{s}}
\def\cpDuration{T_\mathrm{cp}}
\def\subcarrierSpacing{F_\mathrm{s}}

\def\generalIndex{\tau}

\def\numBits{B}
\def\numPayloadBits{B_\text{tot}}
\def\numPolys{P}
\def\polyIndex{p}
\def\numPad{B_\mathrm{pad}}

\def\channelLength{L}
\def\channelIndex{l}
\def\channelIR{\vecsym{h}}
\def\channelGain{g}
\def\channelCoefficient{h}
\def\effLength{L_\mathrm{e}}
\def\effIndex{l_\mathrm{e}}
\def\totLength{L_\mathrm{t}}
\def\totIndex{l_\mathrm{t}}
\def\delay{\tau}
\def\noiseElem{w}

\def\ebno{E_\mathrm{b}/N_0}

\def\subcarrierGain{H}

\def\cfo{\Delta f}
\def\cfoEst{\widehat{\cfo}}

\def\timingOffset{\Delta t}
\def\residualOffset{\delta}
\def\timingOffsetEst{\widehat{\timingOffset}}
\def\residualOffsetEst{\hat{\residualOffset}}

\def\stepBackSamples{N_\text{back}}
\def\po{\phi}
\def\poEst{\hat{\po}}

\def\overFactor{Q}

\def\fracCFO{\mu}
\def\intCFO{u}

\def\T{T}

\def\intVariable{\varphi}

\def\Hspecial{\mathcal{H}}
\def\capacity{C}
\def\estCapacity{\tilde{C}}
\def\meanCapacity{\bar{C}}

\def\omHat{\omega}
\def\kprime{\zeroIndex^\prime}
\def\Kprime{\numZeros^\prime}

\def\Xwilk{\X_\mathrm{W}}
\def\radiusStar{\radius_\ast}

\def\sync{\text{sync}}
\def\normCorr{\Gamma}
\def\sumCorr{U}
\def\energyCorr{V}
\def\nprime{\idftIndex^\prime}
\def\Nprime{\idftSize^\prime}
\def\percentage{\lambda}
\def\syncSample{\generalIndex_\text{sync}}
\def\rxTemplate{\tilde{\vecsym{v}}}
\def\templateMat{\vecsym{T}}
\def\numZerosTM{\numZeros_{\text{tm}}}

\def\channelEst{\hat{\subcarrierGain}}
\def\equalizer{F}
\def\imActiveCarriers{I}
\def\packetDuration{T_\mathrm{packet}}

\acrodef{WSN}{wireless sensor network}
\acrodef{USRP}{universal software radio peripheral}
\acrodef{SN}{sensor node}
\acrodef{FC}{fusion center}
\acrodef{MAC}{multiple-access channel}
\acrodef{FL}{federated learning}
\acrodef{ED}{edge device}
\acrodef{CS}{compressed sensing}
\acrodef{ES}[BS]{base station}
\acrodef{DCN}{data center network}
\acrodef{RIS}{reconfigurable intelligent surfaces}
\acrodef{IMC}{in-memory computing}
\acrodef{FPGA}{field-programmable gate array}
\acrodef{SDR}{software-defined radio}
\acrodef{PS}{processing system}
\acrodef{SS}{soft synchronization}
\acrodef{IQ}{in-phase/quadrature}
\acrodef{IP}{intellectual property}
\acrodef{DMA}{direct-memory access}
\acrodef{RAM}{random access memory}
\acrodef{CC}{companion computer}
\acrodef{FEE}{function estimation error}
\acrodef{MSK}{minimum-shift keying}
\acrodef{TDMA}{time-domain multiple access}
\acrodef{PLNC}{physical-layer network coding}
\acrodef{UAV}{unmanned aerial vehicle}
\acrodef{LoRa}{Long-Range}
\acrodef{DC}{direct-current}
\acrodef{DAC}{digital-to-analog converter}
\acrodef{ADC}{anlog-to-digital converter}
\acrodef{CS}{complementary sequence}
\acrodef{GCP}{Golay complementary pair}
\acrodef{ANF}{algebraic normal form}
\acrodef{AACF}{aperiodic auto-correlation function}
\acrodef{AACFs}{aperiodic auto-correlation functions}
\acrodef{RM}{Reed-Muller}
\acrodef{MOZ}{modulation on zeros}
\acrodef{MOCZ}{modulation on conjugate-reciprocal zeros}
\acrodef{BMOCZ}{binary MOCZ}
\acrodef{BMOCZabstract}[BMOCZ]{binary modulation on conjugate-reciprocal zeros}
\acrodef{JBMOCZ}[J-BMOCZ]{jutted BMOCZ}
\acrodef{dizet}[DiZeT]{direct zero-testing}

\acrodef{PUCCH}{physical uplink control channel}
\acrodef{PRACH}{physical random access channel}

\acrodef{OBO}{output-power back-off}
\acrodef{ACLR}{adjacent-channel-leakage ratio}

\acrodef{LDPC}{low-density parity check}

\acrodef{PDF}{probability density function}
\acrodef{CDF}{cumulative distribution function}
\acrodef{CCDF}{complementary cumulative distribution function}

\acrodef{TBMA}{type-based multiple access}

\acrodef{MSFE}{mean-squared function error}
\acrodef{FEE}{function-estimation error}
\acrodef{CER}{computation error rate}
\acrodef{BCER}{block-computation error rate}
\acrodef{CFO}{carrier frequency offset}
\acrodef{TO}{timing offset}
\acrodef{PO}{phase offset}
\acrodef{RSSI}{received signal strength  information}

\acrodef{STLC}{space-time line code}
\acrodef{CCI}{co-channel interference}
\acrodef{CSIT}[CSIT]{\ac{CSI} at the transmitter}
\acrodef{CSIR}[CSIR]{\ac{CSI} at the receiver}
\acrodef{MIMO}{multiple-input multiple-output}
\acrodef{PC}{phase correction}
\acrodef{ZF}{zero-forcing}
\acrodef{ANOVA}{analysis of variance}

\acrodef{PCA}{principal component analysis}
\acrodef{TIG}{Technical Interest Group}

\acrodef{FSK}{frequency-shift keying}
\acrodef{PPM}{pulse-position modulation}
\acrodef{PAM}{pulse-amplitude modulation}

\acrodef{MRC}{maximum-ratio combining}
\acrodef{HP}{hard-coded participation}
\acrodef{HPA}{hard-coded participation with absentees}
\acrodef{SP}{soft-coded participation}
\acrodef{SIC}{successive interference cancellation}
\acrodef{FSK-MV}{\ac{FSK}-based \ac{MV}}
\acrodef{RF}{radio-frequency}
\acrodef{MF}{matched filter}
\acrodef{PPM}{pulse-position modulation}
\acrodef{CSK}{chirp-shift keying}
\acrodef{PPM-MV}[PPM-MV]{\ac{PPM}-based \ac{MV}}
\acrodef{DFT-s-OFDM}{discrete Fourier transform-spread orthogonal frequency division multiplexing}
\acrodef{SC}{single-carrier}
\acrodef{SGD}{stochastic gradient descent}
\acrodef{signSGD}{sign stochastic gradient descent}

\acrodef{SL}{split learning}
\acrodef{SNR}{signal-to-noise ratio}
\acrodef{RMSE}{root-mean-squared error}
\acrodef{OFDM}{orthogonal frequency division multiplexing}
\acrodef{DFT}{discrete Fourier transform}
\acrodef{PSK}{phase-shift keying}
\acrodef{QAM}{quadrature amplitude modulation}
\acrodef{QPSK}{quadrature phase-shift keying}
\acrodef{PMEPR}{peak-to-mean envelope power ratio}
\acrodef{BER}{bit error rate}
\acrodef{SNR}{signal-to-noise ratio}
\acrodef{PSD}{power spectral density}
\acrodef{SE}{spectral efficiency}
\acrodef{CP}{cyclic prefix}
\acrodef{AWGN}{additive white Gaussian noise}
\acrodef{CFR}{channel frequency response}
\acrodef{CIR}{channel impulse response}
\acrodef{MMSE}{minimum mean-squared error}
\acrodef{LMMSE}{linear minimum mean-squared error}
\acrodef{BPSK}{binary phase shift keying}
\acrodef{BPSK}{quadrature phase shift keying}
\acrodef{BLER}{block error rate}
 \acrodef{ML}{maximum likelihood}
\acrodef{PHY}{physical layer}
\acrodef{PA}{power amplifier}
\acrodef{IDFT}{inverse discrete Fourier transform}
\acrodef{DoF}{degrees-of-freedom}
\acrodef{IoT}{Internet of Things}
\acrodef{mMTC}{massive machine-type communication}
\acrodef{URLLC}{ultra-reliable low-latency communication}
\acrodef{FDE}{frequency-domain equalization}
\acrodef{RF}{radio-frequency}
\acrodef{IM}{index modulation}
\acrodef{MF}{matched filter}
\acrodef{PPM}{pulse-position modulation}

\acrodef{MSE}{mean-squared error}
\acrodef{MRT}{maximum-ratio transmission}
\acrodef{ERC}{equal-ratio combining}
\acrodef{BAA}{broadband analog aggregation}
\acrodef{OBDA}{one-bit broadband digital aggregation}
\acrodef{FEEL}{federated edge learning}
\acrodef{FL}{federated learning}
\acrodef{UL}{uplink}
\acrodef{UCI}{uplink control information}
\acrodef{OAC}{over-the-air computation}
\acrodef{TCI}{truncated-channel inversion}
\acrodef{MV}{majority vote}
\acrodef{CNN}{convolution neural network}
\acrodef{ReLU}{rectified-linear unit}
\acrodef{CSI}{channel state information}
\acrodef{PAPR}{peak-to-average power ratio}
\acrodef{SC}{single-carrier}
\acrodef{iid}[IID]{independent and identically distributed}
\acrodef{RMS}{root-mean-square}
\acrodef{4G}{fourth generation}
\acrodef{5G}{Fifth Generation}
\acrodef{6G}{Sixth Generation}
\acrodef{NR}{New Radio}
\acrodef{LTE}{Long-Term Evolution}
\acrodef{OFDMA}{orthogonal frequency division multiple access}
\acrodef{HARQ}{hybrid automatic repeat request}
\acrodef{D2D}{Device-to-Device}
\acrodef{NOMA}{non-orthogonal multiple access}
\acrodef{OMA}{orthogonal multiple access}
\acrodef{IMT}{International Mobile Telecommunications}
\acrodef{ITU}{International Telecommunication Union}

\acrodef{PDP}{power-delay profile}
\acrodef{TBMA}{type-based multiple access}
\acrodef{ISI}{intersymbol interference}
\acrodef{MLSE}{maximum likelihood sequence estimator}
\acrodef{LTI}{linear time-invariant}
\acrodef{ISAC}{integrated sensing and communication}
\acrodef{DL}{deep learning}
\acrodef{AE}{autoencoder}
\acrodef{AEs}{autoencoders}
\acrodef{MLP}{multi-layer perceptron}
\acrodef{DTFT}{discrete-time Fourier transform}
\acrodef{IDTFT}{inverse discrete-time Fourier transform}
\acrodef{FFT}{fast Fourier transform}
\acrodef{DPSK}{differential phase-shift keying}

\acrodef{ICI}{intercoefficient interference}
\acrodef{CPC}{cyclically permutable code}
\acrodef{ACPC}{affine CPC}

\acrodef{TM}{time mapping}
\acrodef{FM}{frequency mapping}
\acrodef{TFM}{time-frequency mapping}

\acrodef{LLR}{log-likelihood ratio}
\acrodef{PLLR}{pseudo-log-likelihood ratio}

\acrodef{TX}{transmitted}
\acrodef{RX}{received}
\acrodef{DCSK}{differential chaos-shift keying}

\acrodef{I}{in-phase}
\acrodef{Q}{quadrature}

\def\reviewColor{\color{black}}

\begin{document}
	
	\title{
		{Jutted BMOCZ for Non-Coherent OFDM}\\
		\thanks{Parker Huggins was with the Department of Electrical Engineering, University of South Carolina, Columbia, SC 29206 USA. He is now with the Department of Electrical and Computer Engineering, New York University, Brooklyn, NY 11201 USA (e-mail: parker.huggins@nyu.edu).}
		\thanks{Alphan~\c{S}ahin is with the Department of Electrical Engineering, University of South Carolina, Columbia, SC 29206 USA (e-mail: asahin@mailbox.sc.edu).}
		\thanks{This paper was presented in part at the IEEE International Symposium on Personal, Indoor, and Mobile Radio Communications 2025~\cite{huggins2025fourier}.}	
		\thanks{This article has supplementary downloadable material available at https://github.com/parkerkh/Jutted-BMOCZ, provided by the authors.}
		\author{Parker~Huggins,~\IEEEmembership{Graduate Student Member,~IEEE} and Alphan~\c{S}ahin,~\IEEEmembership{Member,~IEEE}} 
		\thanks{This work has been supported by the National Science Foundation through the award CNS-2438837.}
	}
	
	\maketitle
	
	\begin{abstract}  	
		In this work, we propose a zero constellation for \ac{BMOCZabstract}, called \ac{JBMOCZ}, and study its application to non-coherent \ac{OFDM}. With \ac{JBMOCZ}, we introduce rotational asymmetry to the zero constellation for Huffman \ac{BMOCZabstract}, which removes ambiguity at the receiver under a uniform rotation of the zeros. The asymmetry is controlled by the magnitude of ``jutted" zeros and enables the receiver to estimate zero rotation using a simple cross-correlation. The proposed method, however, leads to a natural trade-off between asymmetry and zero stability. Accordingly, we introduce a reliability metric to measure the stability of a polynomial's zeros under an additive perturbation of the coefficients, and we apply the metric to optimize the \ac{JBMOCZ} zero constellation parameters. We then combine the advantages of \ac{JBMOCZ} and Huffman \ac{BMOCZabstract} to design a hybrid waveform for \ac{OFDM} with \ac{BMOCZabstract} (\ac{OFDM}-\ac{BMOCZabstract}). {\reviewColor The pilot-free waveform enables blind synchronization/detection and has a fixed peak-to-average power ratio that is independent of the message.} Finally, we assess the proposed scheme through simulation and demonstrate non-coherent \ac{OFDM}-\ac{BMOCZabstract} using low-cost software-defined radios. 
	\end{abstract}
	
	
	\acresetall
	\section{Introduction} \label{sec:intro}
	
	\acused{OFDM}
	
	\lettrine{O}{rthogonal} frequency division multiplexing (OFDM) is ubiquitous in modern wireless networks. The widespread adoption of \ac{OFDM} is driven by its flexibility in resource allocation, its robustness to frequency-selective fading, and its compatibility with multi-antenna and multi-user systems. Together with its implementation simplicity, these properties have established \ac{OFDM} as the dominant waveform in the standards, including in ongoing work towards 6G~\cite{3gppTD1,3gppTD2}. Yet, \ac{OFDM} often employs coherent modulations, such~as \ac{QAM}, which necessitate pilot overhead to estimate the channel and equalize the received data symbols. This overhead can comprise a significant fraction of the total packet for small payloads. In such settings, \emph{non-coherent} detection is an attractive alternative, where the data are decoded without knowledge of the instantaneous \ac{CSI}.
	
	There are various non-coherent \ac{OFDM} schemes proposed in the literature. Principal among these is non-coherent \ac{OFDM} with \ac{IM}~\cite{choi2017non}, where the bits are encoded into the indices of the active subcarriers, which can be viewed as a generalization of \ac{FSK}. Many works have studied extensions to non-coherent \ac{OFDM}-\ac{IM}, such as \ac{PAPR} reduction~\cite{gopi2020optimized}, code design~\cite{fazeli2019coded}, and receive diversity~\cite{fazeli2020generalized}. Another prominent class of non-coherent \ac{OFDM} schemes are those based on differential encodings. In~\cite{chenhu2020non}, for example, the authors propose a differential modulation for \ac{MIMO}-\ac{OFDM} that exploits the beamforming capabilities of the base station to enable low-complexity, non-coherent detection in the downlink. Similarly, in~\cite{van2020deep}, \ac{DL} is employed for non-coherent detection in the downlink using \ac{DPSK}. Differential \ac{OFDM} has also been studied extensively for underwater acoustic communications, as in~\cite{cai2022joint}, which considers non-coherent \ac{OFDM} with \ac{DCSK}. 
	
	In this study, however, we focus on an unconventional but emerging non-coherent modulation technique, called \ac{MOCZ}, which was first proposed in~\cite{walk2017short}. The core principle of \ac{MOCZ} is to transmit the coefficients of a polynomial whose zeros encode information bits. If the coefficients are modulated in time, as~with a single-carrier waveform, then the convolution of the coefficients with the \ac{CIR} becomes polynomial multiplication in the $z$-domain, which preserves the zeros of the transmitted polynomial, regardless of the \ac{CIR} realization~\cite{walk2019principles}. In this way, the zeros pass through the channel unaltered, enabling the receiver to estimate the transmitted bits without knowledge of the instantaneous \ac{CSI}. However, in practice the coefficients are perturbed by additive noise at the receiver, which non-trivially distorts the zeros. \ac{MOCZ} is especially sensitive to noise, as the zeros of a polynomial are, in general, unstable under an additive perturbation of the coefficients. To combat this sensitivity, the authors in~\cite{walk2019principles} propose a stable \ac{BMOCZ} design, called \emph{Huffman \ac{BMOCZ}}, where the coefficients of each polynomial form a Huffman sequence~\cite{huffman1962generation}. For its numerical stability and favorable autocorrelation properties (see Section~\ref{subsec:preliminaries}), Huffman \ac{BMOCZ} has been the subject of numerous recent studies. In~\cite{walk2020practical}, for example, the authors address the effects of hardware impairments on Huffman \ac{BMOCZ} and propose \acp{CPC} for non-coherent detection under uniform zero rotation. Multi-user \ac{MOCZ} is studied in~\cite{walk2021multi} and~\cite{eren2025non}. The authors in~\cite{siddiqui2023spectrally} apply faster-than-Nyquist signaling to improve the spectral efficiency of \ac{MOCZ}. In~\cite{sahin2024over}, Huffman \ac{BMOCZ} is utilized for non-coherent, over-the-air majority vote computation. In~\cite{sun2023non} and~\cite{ott2025binary}, receive diversity and soft-decision decoding are studied, respectively.
	
	To our knowledge, besides our preliminary work in~\cite{huggins2024optimal}, multi-carrier \ac{MOCZ} has not been considered in the literature. The primary reason for this fact is that \ac{MOCZ} was proposed in~\cite{walk2019principles} for communication over unknown multi-path channels, which are modeled via convolutions in the time domain. Indeed,~\cite{walk2019principles} and~\cite{walk2020practical} consider a \emph{one-shot} scenario, where a single polynomial is transmitted whose duration may even be shorter than the delay spread. In this regime, single-carrier \ac{MOCZ} is attractive. However, we seek to extend \ac{MOCZ} to more conventional communication scenarios, such as uplink \ac{OFDMA}, by~exploiting the unique properties of \ac{BMOCZ} \emph{together} with the flexibility of \ac{OFDM}. Our specific motivation for studying \ac{OFDM}-\ac{BMOCZ} is driven by three factors: 1) its near-optimal \ac{PAPR} with Huffman sequences, which is in contrast to the high \ac{PAPR} of single-carrier Huffman \ac{BMOCZ}~\cite{sasidharan2024alternative}; 2) its support for continuous transmission with a \ac{CP}, which avoids the null between adjacent polynomials required for single-carrier \ac{MOCZ}~\cite{walk2019principles,walk2020practical,walk2021multi}; and 3) its capacity to multiplex users over orthogonal time and frequency resources. Of course, integrating \ac{BMOCZ} with \ac{OFDM} has a particular disadvantage: it sacrifices the invariance of \ac{BMOCZ} to the \ac{CIR} realization. Accordingly, we investigate techniques to address frequency selectivity for \ac{OFDM}-\ac{BMOCZ}, including blind channel estimation. Our analysis takes into account timing and frequency offsets, as considered in~\cite{walk2020practical}. We note, however, that~\cite{walk2020practical} ignores pulse shaping and the loss of orthogonality under a frequency offset, which we account for in this work.   
	
	\subsection{Contributions} \label{subsec:contributions}
	
	Our principal contribution is \ac{JBMOCZ}, which increases the flexibility of Huffman \ac{BMOCZ} and is applicable to \emph{both} single- and multi-carrier \ac{MOCZ}. However, this work emphasizes \ac{OFDM}, where \ac{JBMOCZ} is utilized to correct timing errors. For reference, we summarize our specific contributions below.
	\begin{itemize} 
		\item We propose \ac{JBMOCZ}, a modified Huffman \ac{BMOCZ} zero constellation featuring a ``jutted" zero pair. The key idea of \ac{JBMOCZ} is to introduce rotational asymmetry to the Huffman \ac{BMOCZ} zero constellation. The asymmetry removes the ambiguity during decoding induced by a uniform rotation of the zeros, which is caused by either a timing or frequency offset. We exploit the asymmetry to derive a method for estimating zero rotation that relies on a simple cross-correlation. Our approach avoids the use of \acp{CPC} proposed for Huffman \ac{BMOCZ} in~\cite{walk2020practical} and thereby offers increased flexibility, including support for soft-decision decoding. In fact, we demonstrate that both the \ac{BER} and \ac{BLER} of polar-coded \ac{JBMOCZ} outperform Huffman \ac{BMOCZ} with a \ac{CPC}.
		\item We introduce a reliability metric to measure the stability of a polynomial's zeros under an additive perturbation of the coefficients. The metric is based on the unique interpretation of a polynomial as a frequency-selective but parallel channel, where the channel capacity corresponds to zero stability. Through several examples, we show that the reliability metric enables one to quantify the stability of individual zeros, as well as polynomials and entire \ac{MOCZ} codebooks. We then apply the reliability metric to optimize the zero constellation parameters using an approach that generalizes to other \ac{MOCZ} schemes.
		\item We design a hybrid \ac{OFDM}-\ac{BMOCZ} waveform enabling time-frequency synchronization and detection, all without pilot overhead. The waveform exploits \ac{JBMOCZ} for timing while otherwise leveraging the zero stability and low \ac{PAPR} of Huffman \ac{BMOCZ}. By design, however, the waveform is fundamentally limited by the coherence bandwidth. To address frequency selectivity, we propose a \ac{BMOCZ}-based method for blind channel estimation that permits implementation with larger bandwidths at the expense of increased \ac{PAPR}. Finally, we assess the proposed waveform through simulation and demonstrate its application using \acp{SDR}.
	\end{itemize}
	
	\emph{Organization:} The remainder of the paper is organized as follows: Section~\ref{sec:system_model} provides the system model, including preliminaries on \ac{BMOCZ} and its integration with \ac{OFDM}; Section~\ref{sec:jbmocz} introduces \ac{JBMOCZ}, our method for estimating zero rotation, and our metric for measuring zero stability; Section~\ref{sec:framework} proposes a hybrid waveform for \ac{OFDM}-\ac{BMOCZ}, analyzes its \ac{PAPR}, and discusses blind channel estimation and soft-decision decoding; Section~\ref{sec:results} presents our simulation results and an \ac{SDR} implementation; Section~\ref{sec:conclusion} concludes the paper with remarks on future work.
	
	\emph{Notation:} The real, complex, and natural numbers are denoted by $\R$, $\C$, and $\N$, respectively, and we denote by $[N]=\{0,1,\hdots,N-1\}$ the least $N$ natural numbers. The binary field of dimension $n$ is denoted by $\binaryField^n$, and the polynomial ring over $\C$ is denoted by $\C[z]$.  We denote the complex conjugate and principal value of $z\in\C$ by $\overline{z}$ and $\Arg(z)$, respectively. Small boldface letters denote vectors; large boldface letters denote matrices. The $\ell_2$-norm of a vector $\vector\in\C^N$ is denoted by $||\vector||_2=\sqrt{\vector\herm\vector}$. The operator $\E[\cdot]$ returns the expected value of its argument, and $\complexnormal(0,\sigma^2)$ denotes the circularly-symmetric complex normal distribution with mean zero and variance $\sigma^2$.
	
	\section{System Model} \label{sec:system_model}
	
	We consider a single-user, single-antenna scenario, where the transmitter and receiver communicate over a wireless multi-path channel. The system employs \ac{OFDM}, and each packet comprises $\numSubcarriers$ active subcarriers and $\numOFDMsymbols$ \ac{OFDM} symbols. Let $\sampleRate$ and $\idftSize>\numSubcarriers$ denote the sampling rate and the \ac{IDFT} size, respectively. Then, the duration of each \ac{OFDM} symbol in the time domain is $\symbolDuration=\idftSize/\sampleRate$, and the spacing between adjacent subcarriers in the frequency domain is $\subcarrierSpacing=1/\symbolDuration$. For a \ac{CP} duration $\cpDuration$, the $\ofdmSymbolIndex$th transmitted \ac{OFDM} symbol is given by  
	\begin{equation} \label{eq:mth_tx_symbol}
		s_\ofdmSymbolIndex(t) = \sum_{\subcarrierIndex=0}^{\numSubcarriers-1} \dataSymbol_{\subcarrierIndex,\ofdmSymbolIndex} \e^{j2\pi \subcarrierIndex\subcarrierSpacing t}\,, \ \ t\in[-\cpDuration,\symbolDuration)\,,
	\end{equation} 
	where $\dataSymbol_{\subcarrierIndex,\ofdmSymbolIndex}\in\C$ is the data symbol on the $\subcarrierIndex$th subcarrier of the $\ofdmSymbolIndex$th \ac{OFDM} symbol. 
	
	During transmission, we assume the channel is quasi-static with $\channelLength$ paths, where $\channelGain_\channelIndex\in\C$ and $\delay_\channelIndex>0$ are the gain and delay of the $\channelIndex$th path, respectively. At the receiver, the oscillator used for downconversion is not synchronized with the oscillator at the transmitter, inducing a \ac{CFO} $\cfo$. Additionally, since the receiver has no knowledge of the \ac{CSI} or the time of transmission, let $\timingOffset>0$ denote a \ac{TO} from the point of initial acquisition, which defines the duration of the received signal $\rxSignal(t)$ residing before the packet (see Fig.~\ref{fig:ofdm_diagram}). Under these conditions, the $\ofdmSymbolIndex$th \ac{OFDM} symbol observed by the receiver has the form
	\begin{equation} \label{eq:rx_signal}
		\rxSignal_\ofdmSymbolIndex(t) = \e^{j2\pi \cfo t} \sum_{\channelIndex=0}^{\channelLength-1} \channelGain_\channelIndex \txSignal_\ofdmSymbolIndex(t-\timingOffset-\delay_\channelIndex) + 
		\noiseElem_\ofdmSymbolIndex(t)\,,
	\end{equation}
	where $\noiseElem_\ofdmSymbolIndex(t)$ is zero-mean \ac{AWGN}. When the \ac{CP} duration exceeds the delay spread of the channel, i.e., when $\cpDuration>\delay_\channelLength$, each subcarrier observes a single complex gain, and the received symbol in~\eqref{eq:rx_signal} can be reexpressed as~\cite{zhan2015new,xie2011exact} 
	\begin{equation} \label{eq:mth_rx_symbol}
		\rxSignal_\ofdmSymbolIndex(t) = \e^{j2\pi \cfo t} \sum_{\subcarrierIndex=0}^{\numSubcarriers-1} \dataSymbol_{\subcarrierIndex,\ofdmSymbolIndex} \subcarrierGain_{\subcarrierIndex,\ofdmSymbolIndex} \e^{j2\pi\subcarrierIndex\subcarrierSpacing(t-\timingOffset)} + \noiseElem_\ofdmSymbolIndex(t)\,,
	\end{equation}
	where $\subcarrierGain_{\subcarrierIndex,\ofdmSymbolIndex}\in\C$ is defined as
	\begin{equation} \label{eq:subcarrier_gain}
		\subcarrierGain_{\subcarrierIndex,\ofdmSymbolIndex} \triangleq \e^{j2\pi \cfo (\symbolDuration+\cpDuration)\ofdmSymbolIndex}
		\sum_{\channelIndex=0}^{\channelLength-1} \channelGain_\channelIndex \e^{-j2\pi\subcarrierIndex\subcarrierSpacing\delay_\channelIndex}\,.
	\end{equation}
	In~\eqref{eq:mth_rx_symbol}, the \ac{CFO} induces a shift in the frequency domain, while the \ac{TO} results in a phase modulation of the data symbols. Although the frequency shift introduced by the \ac{CFO} breaks the orthogonality between the subcarriers, in practical systems $\cfo$ can be minimized through a control loop~\cite[Chapter~8.1.1]{dahlman20205g} or readily compensated using conventional methods, such as that of Schmidl and Cox~\cite{schmidl1997robust,moose1994technique}. In fact, we augment the latter method in Section~\ref{subsec:synch} for \ac{OFDM}-\ac{BMOCZ}, enabling the receiver to estimate the \ac{CFO} and decode header bits with a single \ac{OFDM} symbol. Still, the distortion due to the \ac{TO} is the primary focus of this work; it is addressed in Section~\ref{sec:jbmocz} and a corresponding implementation is described in Section~\ref{subsec:to_correction}. 
	
	\subsection{Preliminaries on Binary MOCZ} \label{subsec:preliminaries}
	
	The principle of \ac{BMOCZ} is to encode information bits into the zeros of the transmitted sequence's $z$-transform. Considering a $\numZeros$-bit message $\bMessage=(\bit_0,\bit_1,\hdots,\bit_{\numZeros-1})\in\binaryField^\numZeros$, the $\zeroIndex$th bit $\bit_\zeroIndex$ is mapped to the zero of a polynomial as 
	\begin{equation} \label{eq:bmocz_encoder}
		\zero_\zeroIndex = 
		\begin{cases}
			\zeroMag_\zeroIndex \e^{j\zeroPhase_\zeroIndex}, & \bit_\zeroIndex=1\\
			\zeroMag_\zeroIndex^{-1} \e^{j\zeroPhase_\zeroIndex}, & \bit_\zeroIndex=0\\
		\end{cases}\,, \ \ \zeroIndex\in[\numZeros]\,,
	\end{equation}  
	where the magnitude $\zeroMag_\zeroIndex>1$ and phase $\zeroPhase_\zeroIndex\in[0,2\pi)$ define the \emph{conjugate-reciprocal} zero pair $\crPair{\zeroIndex}\triangleq\{\zeroMag_\zeroIndex \e^{j\zeroPhase_\zeroIndex},\zeroMag_\zeroIndex^{-1} \e^{j\zeroPhase_\zeroIndex}\}$. Under~\eqref{eq:bmocz_encoder}, each binary message $\bMessage\in\binaryField^\numZeros$ maps to a distinct \emph{zero pattern} $\bZero=[\zero_0,\zero_1,\hdots,\zero_{\numZeros-1}]\trans\in\C^\numZeros$, which uniquely defines the polynomial
	\begin{equation} \label{eq:bmocz_poly}
		\X(z) = \sum_{\zeroIndex=0}^{\numZeros} x_\zeroIndex z^\zeroIndex = x_\numZeros \prod_{\zeroIndex=0}^{\numZeros-1} (z-\zero_\zeroIndex)\,, 
	\end{equation}
	up to the scalar $x_\numZeros\neq0\in\C$. The polynomial $\X(z)$ has $\numZeros+1$ coefficients and is the $z$-transform of the sequence $\txSeq=[x_0,x_1,\hdots,x_\numZeros]\trans\in\C^{\numZeros+1}$. To render~\eqref{eq:bmocz_poly} one-to-one, the leading coefficient $x_\numZeros$ is selected so that $||\txSeq||_2^2=\numZeros+1$. The set of all normalized polynomial sequences then defines a \emph{codebook}, denoted by $\polyCodebook$, with cardinality $|\polyCodebook|=2^\numZeros$.
	
	The merit of single-carrier \ac{BMOCZ} is that, regardless of the \ac{CIR} realization, the zero structure of $\X(z)$ is preserved at the receiver~\cite{walk2017short,walk2019principles}. For example, consider a single-carrier waveform modulated by $\txSeq\in\polyCodebook$. After matched filtering, assuming no \ac{TO} or \ac{CFO}, the received sequence $\rxSeq\in\C^{\totLength}$ is given by $\rxSeq=\txSeq\ast\channelIR+\noiseSeq$, where $\channelIR\in\C^{\effLength}$ is the \ac{CIR} of effective length $\effLength\geq1$, and $\noiseSeq\in\C^{\totLength}$ is \ac{AWGN} with $\totLength=\numZeros+\effLength$. In the $z$-domain, the convolution of $\txSeq$ with the \ac{CIR} becomes polynomial multiplication as
	\begin{align} \label{eq:rx_poly}
		\Y(z) &= \X(z)\H(z) + \W(z) \notag \\
		&= x_\numZeros h_{\effLength-1} 
		\prod_{\zeroIndex=0}^{\numZeros-1}(z-\zero_\zeroIndex) 
		\prod_{\effIndex=0}^{\effLength-2}(z-\channelZero_{\effIndex}) \notag\\
		&\quad +\, w_{\totLength-1} 
		\prod_{\totIndex=0}^{\totLength-2}(z-\noiseZero_{\totIndex})\,,
	\end{align}
	where $\H(z)$ and $\W(z)$ are the $z$-transforms of $\channelIR$ and $\noiseSeq$, respectively. Observe that the received polynomial in~\eqref{eq:rx_poly} has $\totLength-1$ zeros, i.e., the $\numZeros$ data zeros $\{\zero_\zeroIndex\}$ and the $\effLength-1$ zeros $\{\channelZero_{\effIndex}\}$ introduced by the channel, all perturbed by noise. Since the data zeros remain approximately as roots of the received polynomial (exactly in the noiseless case), \ac{BMOCZ} enables non-coherent detection at the receiver. In particular, a simple \emph{\ac{dizet}} decoder is introduced in~\cite{walk2019principles}, where $\Y(z)$ is evaluated (tested) at the zeros in $\crPair{\zeroIndex}$, and the $\zeroIndex$th message bit $\bit_\zeroIndex$ is estimated as
	\begin{equation} \label{eq:hd_dizet}
		\estimatedBit_\zeroIndex = 
		\begin{cases}
			1, & |\Y(\zeroMag_\zeroIndex\e^{j\zeroPhase_\zeroIndex})| < \zeroMag_\zeroIndex^{\totLength-1} |\Y(\zeroMag_\zeroIndex^{-1}\e^{j\zeroPhase_\zeroIndex})|\\
			0, & \text{otherwise}
		\end{cases}\,, \ \ \zeroIndex\in[\numZeros]\,.
	\end{equation}
	The evaluation $|\Y(\zeroMag_\zeroIndex^{-1}\e^{j\zeroPhase_\zeroIndex})|$ is scaled by $\zeroMag_\zeroIndex^{\totLength-1}$ to balance the growth of $|\Y(z)|$ for $|z|>1$ and thereby ensure a fair comparison between $\zero_\zeroIndex\raisedOne\triangleq\zeroMag_\zeroIndex\e^{j\zeroPhase_\zeroIndex}$ and $\zero_\zeroIndex\raisedZero\triangleq\zeroMag_\zeroIndex^{-1}\e^{j\zeroPhase_\zeroIndex}$. See~the analysis in~\cite[Section~III-D]{walk2019principles} for more details.
	
	It is worth noting that the \ac{dizet} decoder is sensitive to zero perturbation under noise, since movement of the zeros can induce erroneous evaluations in~\eqref{eq:hd_dizet} and thus a misdetection. Therefore, as discussed in~\cite[Section~IV]{walk2019principles} and Section~\ref{subsec:optimization} of this work, it is important to optimize the zero constellation parameters $\{\zeroMag_\zeroIndex\}$ and $\{\zeroPhase_\zeroIndex\}$ to ensure robustness against noise. Furthermore, the \ac{dizet} decoder assumes knowledge of the effective \ac{CIR} length via $\totLength=\numZeros+\effLength$. For single-carrier \ac{BMOCZ}, then, it is necessary to estimate $\effLength$ at the receiver to limit the number of channel zeros introduced in~\eqref{eq:rx_poly} and ensure proper scaling in~\eqref{eq:hd_dizet}. With \ac{OFDM}-\ac{BMOCZ}, however, the coefficients are multiplied pointwise with channel gains (not convolved), and hence no zeros are introduced to the transmitted polynomials. The distortion of the zeros under pointwise multiplication depends on the frequency selectivity of the channel and the mapping of the polynomial coefficients to time-frequency resources (see Section~\ref{subsec:bmocz_ofdm}).
	
	We now define the \ac{AACF} of a sequence $\txSeq\in\C^{\numZeros+1}$, which is of importance moving forward. 
	\begin{definition}
		The \ac{AACF} of $\txSeq=[x_0,x_1,\hdots,x_\numZeros]\trans\in\C^{\numZeros+1}$ is the complex-valued function $\A(z) \triangleq \sum_{\autoIndex=-\numZeros}^{\numZeros}a_\autoIndex z^\autoIndex$, where the coefficient $a_\autoIndex\in\C$ is defined as 
		\begin{equation} \label{eq:auto_coeffs}
			a_\autoIndex \triangleq 
			\begin{cases} 
				\sum_{i=0}^{\numZeros-\autoIndex} \overline{x_i} x_{i+\autoIndex}, & 0 \leq \autoIndex \leq \numZeros\\
				\sum_{i=0}^{\numZeros+\autoIndex} x_i \overline{x_{i-\autoIndex}}, & -\numZeros \leq \autoIndex < 0\\
				0, & \textup{otherwise}
			\end{cases}\,.
		\end{equation}
	\end{definition}
	\noindent 
	For \ac{BMOCZ}, the \ac{AACF} of $\txSeq\in\polyCodebook$ can be computed directly from the zeros of $\X(z)$ as 
	\begin{align} \label{eq:auto_poly}
		\A(z) &= \sum_{\autoIndex=-\numZeros}^{\numZeros} a_\autoIndex z^\autoIndex = \X(z)\overline{\X(1/\overline{z})} \notag\\
		&= z^{-\numZeros} x_\numZeros\overline{x_0} \prod_{\zeroIndex=0}^{\numZeros-1} (z-\zero_\zeroIndex) \prod_{\zeroIndex=0}^{\numZeros-1} (z-1/\overline{\zero_\zeroIndex})\,.
	\end{align}
	It follows that $\A(z)$ in~\eqref{eq:auto_poly} is independent of the zero pattern, and hence each sequence $\txSeq\in\polyCodebook$ has an \emph{identical} \ac{AACF}. By leveraging this property and the underlying zero structure of Huffman sequences~\cite{huffman1962generation}, \emph{Huffman \ac{BMOCZ}} is proposed in~\cite{walk2019principles}, where the zeros in~\eqref{eq:bmocz_encoder} lie uniformly along two circles of radii $\radius$ and $\radius^{-1}$, i.e., $\zeroMag_\zeroIndex=\radius>1$ and $\zeroPhase_\zeroIndex=2\pi\zeroIndex/\numZeros$, $\zeroIndex\in[\numZeros]$. A polynomial $\Xhuffman(z)$ generated with this encoding is called a \emph{Huffman polynomial}, since its coefficients $\txSeq\in\polyCodebook_{\text{H}}$ form a Huffman sequence. Notably, Huffman sequences have an impulsive \ac{AACF} of the form 
	\begin{equation} \label{eq:huffman_auto}
		\Ahuffman(z)=(\numZeros+1)(-\etaH z^{-\numZeros} + 1 - \etaH z^\numZeros)\,,
	\end{equation}
	where $\etaH\triangleq1/(\radius^\numZeros+\radius^{-\numZeros})$~\cite{walk2019principles}. The impulse-like \ac{AACF} of Huffman sequences makes them ideal for applications in radar~\cite{ackroyd1970design,dehkordi2023integrated}. Furthermore, as demonstrated through \ac{BER} simulations and polynomial perturbation analysis, the zeros of Huffman polynomials are stable under noise~\cite{walk2019principles}, a desirable property for the \ac{dizet} decoder in~\eqref{eq:hd_dizet}. A disadvantage of Huffman sequences in the time domain, however, is their large~\ac{PAPR}, as the first and last coefficients alone carry almost half of the sequence energy~\cite{walk2019principles,walk2020practical,sasidharan2024alternative}.
	
	\subsection{Subcarrier Mappings for OFDM-BMOCZ} \label{subsec:bmocz_ofdm}
	
	Assume that we transmit $\numPolys$ polynomials, each of degree $\numZeros$, and let $x_{\zeroIndex,\polyIndex}$ denote the $\zeroIndex$th coefficient of the $\polyIndex$th polynomial for $\zeroIndex\in[\numZeros+1]$ and $\polyIndex\in[\numPolys]$. In this work, the coefficients are mapped to the subcarriers using \ac{FM}~\cite{huggins2024optimal}, where $\numSubcarriers=\numZeros+1$ subcarriers and $\numOFDMsymbols=\numPolys$ \ac{OFDM} symbols are resourced, and the data symbol $\dataSymbol_{\subcarrierIndex,\ofdmSymbolIndex}$ in~\eqref{eq:mth_tx_symbol} is selected as $\dataSymbol_{\subcarrierIndex,\ofdmSymbolIndex} = x_{\subcarrierIndex,\ofdmSymbolIndex}$. Equivalently, the $\ofdmSymbolIndex$th \ac{OFDM} symbol is given~by the $\ofdmSymbolIndex$th transmit polynomial $\X_\ofdmSymbolIndex(z)$ evaluated along the unit~circle, i.e., $\txSignal_\ofdmSymbolIndex(t)=\sum_{\subcarrierIndex=0}^{\numZeros}x_{\subcarrierIndex,\ofdmSymbolIndex}\e^{j2\pi \subcarrierIndex\subcarrierSpacing t}=\X_\ofdmSymbolIndex(\e^{j2\pi\subcarrierSpacing t})$. For Huffman \ac{BMOCZ}, \ac{FM} yields very low \ac{PAPR} due to the impulsive \ac{AACF} of Huffman sequences [see Fig.~\ref{fig:papr_space}(a)]. However, \ac{FM} is sensitive to frequency selectivity, as the $\ofdmSymbolIndex$th received polynomial (noiseless) $\Y_\ofdmSymbolIndex(z)=\sum_{\subcarrierIndex=0}^{\numZeros}x_{\subcarrierIndex,\ofdmSymbolIndex}\subcarrierGain_{\subcarrierIndex,\ofdmSymbolIndex}z^\subcarrierIndex$ has coefficients given by a pointwise product over \emph{frequency}. \ac{FM} is thus suited for narrowband, small payload applications, such as transmitting \ac{UCI} within an \ac{OFDMA} framework. \Ac{TM} is an alternative to \ac{FM} with $\numSubcarriers=\numPolys$ subcarriers and $\numOFDMsymbols=\numZeros+1$ \ac{OFDM} symbols, where $\dataSymbol_{\subcarrierIndex,\ofdmSymbolIndex}$ in~\eqref{eq:mth_tx_symbol} is set to $\dataSymbol_{\subcarrierIndex,\ofdmSymbolIndex} = x_{\ofdmSymbolIndex,\subcarrierIndex}$. \ac{TM} is robust against frequency selectivity, as the $\subcarrierIndex$th received polynomial $\Y_\subcarrierIndex(z)=\sum_{\ofdmSymbolIndex=0}^{\numZeros}x_{\ofdmSymbolIndex,\subcarrierIndex}\subcarrierGain_{\subcarrierIndex,\ofdmSymbolIndex}z^\ofdmSymbolIndex$ has coefficients given by a pointwise product over \emph{time}, and $\subcarrierGain_{\subcarrierIndex,\ofdmSymbolIndex}$ is roughly constant in $\ofdmSymbolIndex$ by our assumption that the channel is quasi-static. The disadvantage of \ac{TM} is its high \ac{PAPR} with Huffman sequences (see \cite[Fig.~7]{huggins2024optimal}). Nevertheless, we show in Section~\ref{subsec:chest} that \ac{TM} can be exploited for blind channel estimation while simultaneously conveying information to the receiver. We refer the reader to~\cite[Fig.~2]{huggins2024optimal} for visualizations of \ac{TM} and \ac{FM}.
	
	\subsection{Problem Statement: Zero Rotation under Timing Error} \label{subsec:problem}
	
	After coarse synchronization to the first \ac{OFDM} symbol, it is conventional to ``step back" a few samples and ensure that the \ac{DFT} window excludes samples from the subsequent \ac{OFDM} symbol, which would lead to \ac{ISI}~\cite{yang2000timing}. In practice, then, there persists a residual \ac{TO} after coarse synchronization. The residual \ac{TO} depends on both the synchronization error and the \ac{OFDM} symbol index $\ofdmSymbolIndex$, as receiver clock drift results in gradual sampling misalignments~\cite{yang2000timing},\,\cite{stantchev2002time},\,\cite{lee2011joint}. Let $\timingOffsetEst>0$ denote a coarse estimate of the \ac{TO}, and assume a step back of $\stepBackSamples\in\N$~samples. We express the residual \ac{TO} as
	\begin{equation}
		\residualOffset_\ofdmSymbolIndex = \underbrace{\timingOffset-\left(\timingOffsetEst-\frac{\stepBackSamples}{\sampleRate}\right)}_{\text{synchronization error}}+\underbrace{\epsilon_\ofdmSymbolIndex}_{\substack{\text{clock}\\\text{drift}}}\,,
	\end{equation}
	where $\epsilon_\ofdmSymbolIndex\in\R$ is the aggregate clock drift by the symbol $\ofdmSymbolIndex$ (see Fig.~\ref{fig:ofdm_diagram}). With \ac{FM}, after synchronization and \ac{CP} removal, and assuming no \ac{CFO} or noise for the ease of exposition, the~$\idftIndex$th sample of the $\ofdmSymbolIndex$th \ac{OFDM} symbol in~\eqref{eq:mth_rx_symbol} becomes
	\begin{equation} \label{eq:mth_symbol_simplified}
		\rxSignalMod_\ofdmSymbolIndex(\idftIndex/\sampleRate) = \sum_{\subcarrierIndex=0}^{\numSubcarriers-1} y_{\subcarrierIndex,\ofdmSymbolIndex} \e^{-j2\pi\frac{\subcarrierIndex}{\idftSize}\residualOffset_\ofdmSymbolIndex\sampleRate}\e^{j2\pi\frac{\subcarrierIndex}{\idftSize}\idftIndex}\,, \ \ \idftIndex\in[\idftSize]\,,
	\end{equation}
	where $y_{\subcarrierIndex,\ofdmSymbolIndex} \triangleq x_{\subcarrierIndex,\ofdmSymbolIndex}\subcarrierGain_{\subcarrierIndex,\ofdmSymbolIndex}$. Taking the $\idftSize$-point \ac{DFT} of \eqref{eq:mth_symbol_simplified} to~retrieve the coefficients $\rxSeq_\ofdmSymbolIndex$, we obtain
	\begin{equation} \label{eq:to_effect}
		\tilde{y}_{\subcarrierIndex,\ofdmSymbolIndex} = y_{\subcarrierIndex,\ofdmSymbolIndex} \e^{-j2\pi\frac{\subcarrierIndex}{\idftSize}\residualOffset_\ofdmSymbolIndex\sampleRate}\,, \ \ \subcarrierIndex\in[\numZeros+1]\,.
	\end{equation}
	Then, defining the angle $\po_\ofdmSymbolIndex\triangleq2\pi\residualOffset_\ofdmSymbolIndex\sampleRate/\idftSize$ and transforming \eqref{eq:to_effect} into the $z$-domain, the received polynomial over the $\ofdmSymbolIndex$th \ac{OFDM} symbol is given by 
	\begin{align} \label{eq:zero_rotation}
		\Ymod_\ofdmSymbolIndex(z) &= \sum_{\subcarrierIndex=0}^{\numZeros} y_{\subcarrierIndex,\ofdmSymbolIndex} \e^{-j\po_\ofdmSymbolIndex\subcarrierIndex} z^\subcarrierIndex = \Y_\ofdmSymbolIndex(\e^{-j\po_\ofdmSymbolIndex}z) \notag\\
		&= \e^{-j\po_\ofdmSymbolIndex\numZeros} y_{\numZeros,\ofdmSymbolIndex} \prod_{\zeroIndex=0}^{\numZeros-1} (z - \tilde{\zero}_\zeroIndex\e^{j\po_\ofdmSymbolIndex})\,.
	\end{align}
	In~\eqref{eq:zero_rotation}, the coefficients of $\Y_\ofdmSymbolIndex(z)$ are phase modulated due~to the residual \ac{TO} $\residualOffset_\ofdmSymbolIndex\in[0,\cpDuration]$, which induces a uniform, counterclockwise rotation of the $\numZeros$ received zeros $\tilde{\bZero}$ through the angle $\po_\ofdmSymbolIndex\in[0,2\pi\cpDuration/\symbolDuration]$.
	
	Zero rotation is problematic for Huffman \ac{BMOCZ}, since any rotation of the zeros is ambiguous modulo $\baseAngle\triangleq2\pi/\numZeros$ (see~\cite[Fig.~4(a)]{walk2020practical}). To address this ambiguity, the authors in~\cite{walk2020practical} decompose the rotation $\po$ as $\po=(\intCFO+\fracCFO)\baseAngle$, where $\intCFO\in[\numZeros]$ and $\fracCFO\in[0,1)$ are integer and fractional multiples of the base angle $\baseAngle$, respectively. The fractional rotation $\fracCFO\baseAngle$ is estimated using an oversampled \ac{dizet} decoder, while the integer rotation $\intCFO\baseAngle$, which induces a cyclic permutation of the transmitted message for $\intCFO\geq1$, is determined using an \emph{\ac{ACPC}}. Although effective, the proposed \ac{ACPC} restricts the choice of coding to \acp{CPC} and places numerous constraints on the message length and rate. In this work, towards increasing the flexibility of Huffman \ac{BMOCZ} under zero rotation, we seek to answer the following question: \emph{how can we design a \ac{CPC}-free method for estimating zero rotation and apply it to synthesize \ac{OFDM}-\ac{BMOCZ} waveforms robust against hardware impairments at the physical layer?}
	
	\section{Jutted BMOCZ} \label{sec:jbmocz}
	
	\begin{figure*}[t!]
		\centering
		
		\subfloat[Zero constellation.]{\includegraphics[width=1.725in]{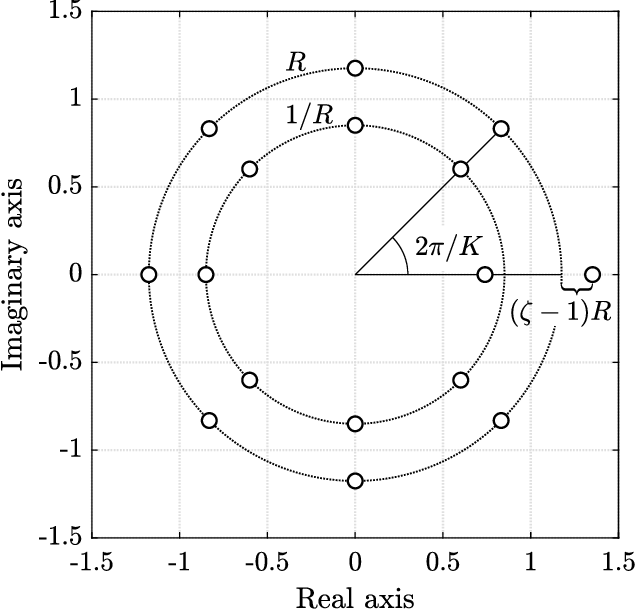}\label{subfig:jbmocz_constellation}}~
		\subfloat[Transmitted zeros.]{\includegraphics[width=1.725in]{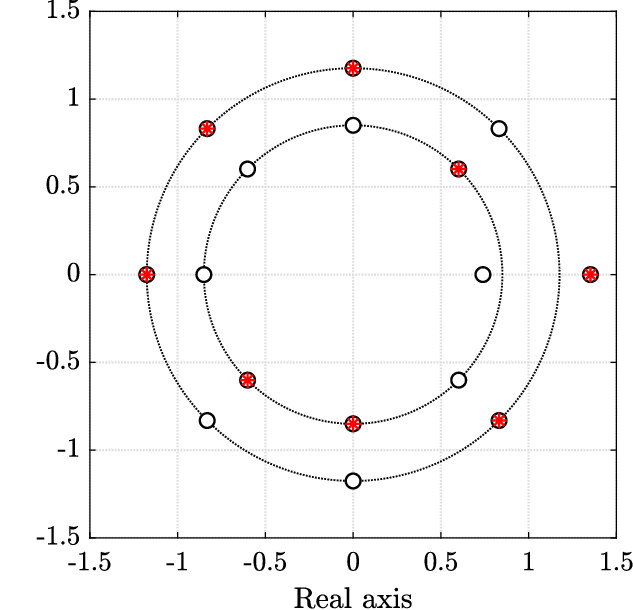}\label{subfig:jbmocz_tx}}~
		\subfloat[Received zeros.]{\includegraphics[width=1.725in]{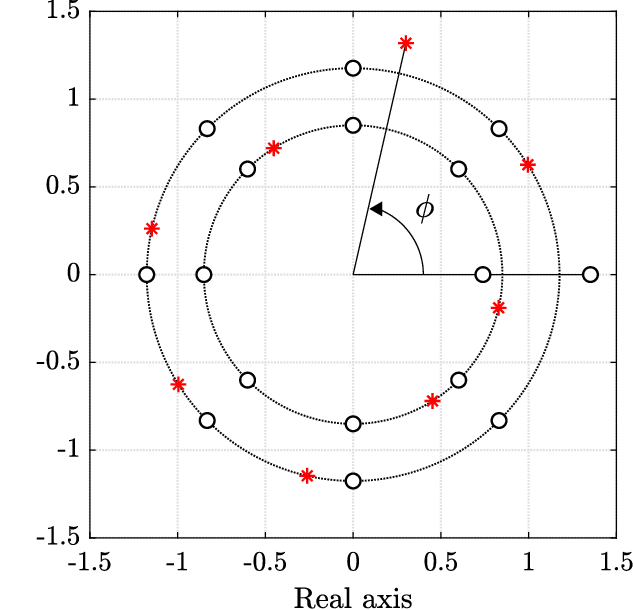}\label{subfig:jbmocz_rx}}~
		\subfloat[Corrected zeros.]{\includegraphics[width=1.725in]{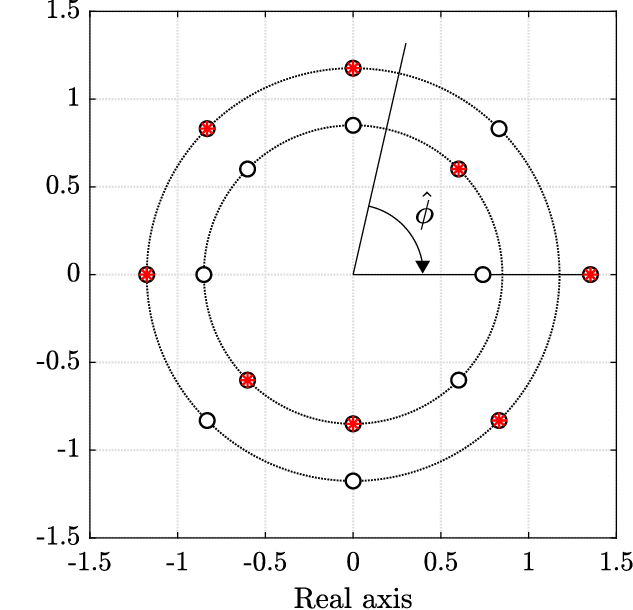}\label{subfig:jbmocz_corrected}}
		
		\caption{Proposed \ac{JBMOCZ} zero pattern for $\numZeros=8$, $\radius=1.176$, and $\asymFactor=1.15$. (a) Full zero constellation with $2\numZeros$ zero positions. (b) Transmitted zeros corresponding to the message $\bMessage=(1,0,1,1,1,0,0,1)$. (c) Received zeros rotated by $\po=(12/7)\baseAngle$ radians. (d) Received zeros after correcting $\po$ via~\eqref{eq:template_corr}, which exactly correspond to the transmitted message.}	
		\label{fig:jbmocz_rotation}
	\end{figure*}
	
	The high-level idea behind \ac{JBMOCZ} is simple: introduce \emph{asymmetry} to the Huffman \ac{BMOCZ} zero constellation to remove the ambiguity under zero rotation. Though there are many ways to inject asymmetry into the zero constellation, perhaps the simplest is to augment the magnitude of a single conjugate-reciprocal zero pair, which is what we consider with \ac{JBMOCZ}. Specifically, we maintain uniform phase separation of the zeros with $\zeroPhase_\zeroIndex=2\pi\zeroIndex/\numZeros$, $\zeroIndex\in[\numZeros]$, but we set
	\begin{equation} \label{eq:new_radius}
		\zeroMag_\zeroIndex = 
		\begin{cases}
			\asymFactor\radius, & \zeroIndex=0\\
			\radius, & \text{otherwise}
		\end{cases}\,, \ \ \zeroIndex\in[\numZeros]\,,
	\end{equation}  
	where $\asymFactor\geq1$ is an \emph{asymmetry factor}. With this encoding, each zero in~\eqref{eq:bmocz_encoder} lies on a circle of radius $\radius$ or $\radius^{-1}$, except for the zero corresponding to the first message bit, $\zero_0$, which is real and lies on a circle of radius $\asymFactor\radius\geq\radius$ or $(\asymFactor\radius)^{-1}\leq\radius^{-1}$. The asymmetry factor controls how far $\zero_0$ ``juts out." In the special case $\asymFactor=1$, $\zeroMag_\zeroIndex=\radius$ for all $\zeroIndex\in[\numZeros]$, and the proposed encoding reduces to Huffman \ac{BMOCZ}. An illustration with $\asymFactor>1$ is provided in Fig.~\ref{fig:jbmocz_rotation}. We discuss in Section~\ref{subsec:template_estimation} how we exploit the rotational asymmetry to estimate zero rotation at the receiver.  
	
	Barring the first zero pair $\crPair{0}=\{\zero_0,1/\overline{\zero_0}\}$, note that the \ac{JBMOCZ} and Huffman \ac{BMOCZ} zero constellations are identical. Hence, in light of~\eqref{eq:auto_poly}, we can express the \ac{AACF} for \ac{JBMOCZ} in terms of the Huffman \ac{BMOCZ} \ac{AACF} in~\eqref{eq:huffman_auto} by removing the factors $(z-\radius)$ and $(z-\radius^{-1})$ and replacing them with the factors $(z-\asymFactor\radius)$ and $(z-\asymFactor^{-1}\radius^{-1})$, respectively. This approach gives
	\begin{align}
		\Ajutted(z) = -\etaJ (\numZeros+1) &\frac{(z^\numZeros-\radius^\numZeros)(z^\numZeros-\radius^{-\numZeros})}{(z-\radius)(z-\radius^{-1})}\notag\\
		&\times (z-\asymFactor\radius)(z-\asymFactor^{-1}\radius^{-1})z^{-\numZeros}\,,
	\end{align} 
	where $\etaJ>0$ scales $\Ajutted(z)$ such that $a_0=\numZeros+1$~[cf.~\eqref{eq:auto_coeffs}]. To~identify $\etaJ$ in terms of $\numZeros$, $\radius$, and  $\asymFactor$, we define
	\begin{align}
		&\X^{(\radius,\asymFactor)}(z) \triangleq \frac{z^\numZeros-\radius^\numZeros}{z-\radius}(z-\asymFactor\radius)\notag\\
		&= - \asymFactor\radius^\numZeros + \sum_{\zeroIndex=1}^{\numZeros-1}(1-\asymFactor)\radius^{\numZeros-\zeroIndex}z^\zeroIndex + z^\numZeros = \sum_{\zeroIndex=0}^{\numZeros} c_\zeroIndex(\radius,\asymFactor)z^\zeroIndex\,,
	\end{align} 
	where 
	\begin{equation}
		c_\zeroIndex(\radius,\asymFactor) = 
		\begin{cases}
			-\asymFactor\radius^\numZeros, & \zeroIndex=0\\
			(1-\asymFactor)\radius^{\numZeros-\zeroIndex}, & \zeroIndex=1,2,\hdots,\numZeros-1\\
			1, & \zeroIndex=\numZeros
		\end{cases}\,.
	\end{equation}
	Then, $\Ajutted(z)=-\etaJ(\numZeros+1)\X^{(\radius,\asymFactor)}(z)\X^{(\radius^{-1},\asymFactor^{-1})}(z)z^{-\numZeros}$, and with~\eqref{eq:auto_coeffs} we get
	\begin{align} \label{eq:eta_inverse}
		&\frac{a_0}{\numZeros+1} = -\etaJ\sum_{\zeroIndex=0}^{\numZeros} c_\zeroIndex(\radius,\asymFactor) c_{\numZeros-\zeroIndex}({\radius^{-1},\asymFactor^{-1}})=1 \notag\\
		\iff&\etaJ^{-1}=-\sum_{\zeroIndex=0}^{\numZeros} c_\zeroIndex(\radius,\asymFactor) c_{\numZeros-\zeroIndex}({\radius^{-1},\asymFactor^{-1}})\,.
	\end{align}
	The terms for $\zeroIndex=0$ and $\zeroIndex=\numZeros$ in~\eqref{eq:eta_inverse} simplify to $-\zeta\radius^\numZeros$ and $-\asymFactor^{-1}\radius^{-\numZeros}$, respectively, while the middle $\numZeros-1$ terms are of the form $(1-\asymFactor)(1-\asymFactor^{-1})\radius^{\numZeros-2\zeroIndex}$, $\zeroIndex\in\{1,2,\hdots,\numZeros-1\}$. Invoking the identity
	\begin{equation}
		\sum_{\zeroIndex=1}^{\numZeros-1}\radius^{\numZeros-2\zeroIndex} = \frac{\radius^{\numZeros-1}-\radius^{-(\numZeros-1)}}{\radius-\radius^{-1}}\,,
	\end{equation}
	we then obtain
	\begin{equation} \label{eq:jutted_eta}
		\etaJ = \frac{1}{\asymFactor\radius^\numZeros + \asymFactor^{-1}\radius^{-\numZeros} - (1-\asymFactor)(1-\asymFactor^{-1}) \frac{\radius^{\numZeros-1}-\radius^{-(\numZeros-1)}}{\radius-\radius^{-1}}}\,.
	\end{equation}	
	Notice that for $\asymFactor=1$, $\etaJ=1/(\radius^\numZeros+\radius^{-\numZeros})=\etaH$, and $\Ajutted(z)$ reduces to $\Ahuffman(z)$ in~\eqref{eq:huffman_auto}. Additionally, evaluating $\Ajutted(z)$ along the unit circle $\{z=\e^{j\omHat} \ | \ \omHat\in[0,2\pi)\}$ gives
	\begin{equation} \label{eq:jutted_psd}
		\Ajutted(\e^{j\omHat}) = \frac{\etaJ}{\etaH} \frac{2\cos(\omHat) - (\asymFactor\radius+\asymFactor^{-1}\radius^{-1})}{2\cos(\omHat)-(\radius+\radius^{-1})} \Ahuffman(\e^{j\omHat})\,,				
	\end{equation}
	where $\Ahuffman(\e^{j\omHat})=(\numZeros+1)[1-2\etaH\cos(\numZeros\omHat)]$. The functions $\Ahuffman(\e^{j\omega})$ and $\Ajutted(\e^{j\omega})$ represent the instantaneous baseband signal power of \ac{FM}-based \ac{OFDM}-\ac{BMOCZ} symbols with Huffman \ac{BMOCZ} and \ac{JBMOCZ}, respectively; they are used in Section~\ref{subsec:papr} to derive expressions for the \ac{PAPR} in terms of the zero constellation parameters $\numZeros$, $\radius$, and $\asymFactor$.
	
	\subsection{Template-based Estimation of Zero Rotation} \label{subsec:template_estimation}
	
	From~\eqref{eq:auto_poly}, we have $\A(\e^{j\omega}) = \X(\e^{j\omega})\overline{\X(\e^{j\omega})} = |X(\e^{j\omega})|^2$. Recalling that $\A(z)$ is fixed, it follows that the magnitude $|\X(\e^{j\omega})|$ is identical for each codeword $\txSeq\in\polyCodebook$. Hence, define $\T(\omega)\triangleq|\X(\e^{j\omHat})|$, which we call a \emph{template transform} (or simply \emph{template}), as it corresponds to the magnitude of the \ac{IDTFT} of each $\txSeq\in\polyCodebook$, i.e., $\T(\omega)=|\sum_{\zeroIndex=0}^{\numZeros}x_\zeroIndex\e^{j\omega\zeroIndex}|$. Since $\T(\omega)$ is common to every codeword, it is known at the receiver and captures the expected ``shape" of $|\Y_\ofdmSymbolIndex(\e^{j\omHat})|$. However, under the uniform zero rotation in~\eqref{eq:zero_rotation}, the receiver observes a frequency-shifted version of $|\Y_\ofdmSymbolIndex(\e^{j\omHat})|$, namely, $|\Ymod_\ofdmSymbolIndex(\e^{j\omHat})|=|\Y_\ofdmSymbolIndex(\e^{j(\omHat-\po_\ofdmSymbolIndex)})|$. Thus, we estimate $\po_\ofdmSymbolIndex$ by correlating $\T(\omega)$ with $|\Ymod_\ofdmSymbolIndex(\e^{j\omHat})|$ as
	\begin{equation} \label{eq:template_corr}
		\poEst_\ofdmSymbolIndex = \arg \max_{\varphi\in[0,2\pi)} \int_{0}^{2\pi} \T(\omHat-\intVariable) |\Ymod_\ofdmSymbolIndex(\e^{j\omHat})| \, \mathrm{d}\omHat\,.
	\end{equation}
	Importantly, for the cross-correlation in~\eqref{eq:template_corr} to have a unique maximum, thus ensuring the estimate of $\po_\ofdmSymbolIndex$ is unambiguous, $\T(\omega)$ must be \emph{aperiodic} over $\omega\in[0,2\pi)$. Equivalently, it is required that the zero constellation be rotationally \emph{asymmetric}. We plot example template transforms for \ac{JBMOCZ} in Fig.~\ref{fig:templates}. As expected, for Huffman \ac{BMOCZ} (i.e., $\asymFactor=1$), the template is periodic with a period equal to the base angle $\baseAngle=2\pi/\numZeros$. For $\asymFactor>1$, the jutted zero introduces a peaky envelope to the sinusoidal Huffman \ac{BMOCZ} template that breaks its periodicity. Indeed, the template transform $\T(\omega)=|\Xjutted(\e^{j\omega})|$ for \ac{JBMOCZ} is well-suited for the cross-correlation in~\eqref{eq:template_corr}, as the product $\T(\omHat-\intVariable)|\Ymod_\ofdmSymbolIndex(\e^{j\omHat})|$ has a maximum when the shift $\intVariable$ aligns the peaks of $\T(\omHat-\intVariable)$ and $|\Ymod_\ofdmSymbolIndex(\e^{j\omHat})|$. Beyond shifting due to zero rotation, however, the received template $|\Ymod_\ofdmSymbolIndex(\e^{j\omHat})|$ is also distorted by fading and noise. For example, with single-carrier \ac{JBMOCZ}, the channel zeros introduced in~\eqref{eq:rx_poly} non-linearly distort $|\Ymod_\ofdmSymbolIndex(\e^{j\omHat})|$. Yet, we show in Section~\ref{subsec:comparison} that, with coding, the template-based estimator still yields good performance for single-carrier \ac{JBMOCZ} (see Fig.~\ref{fig:error_with_coding}). Similarly, for \ac{OFDM}-\ac{BMOCZ} with \ac{FM}, the received template is distorted by frequency-selective fading. The template-based implementation in Section~\ref{subsec:to_correction} thus relies on the narrowband assumption of Section~\ref{subsec:bmocz_ofdm}, or the use of blind channel estimation and equalization, as considered in Section~\ref{subsec:chest}.
	
	\begin{figure}[t!]
		\centering
		\includegraphics[width=3in]{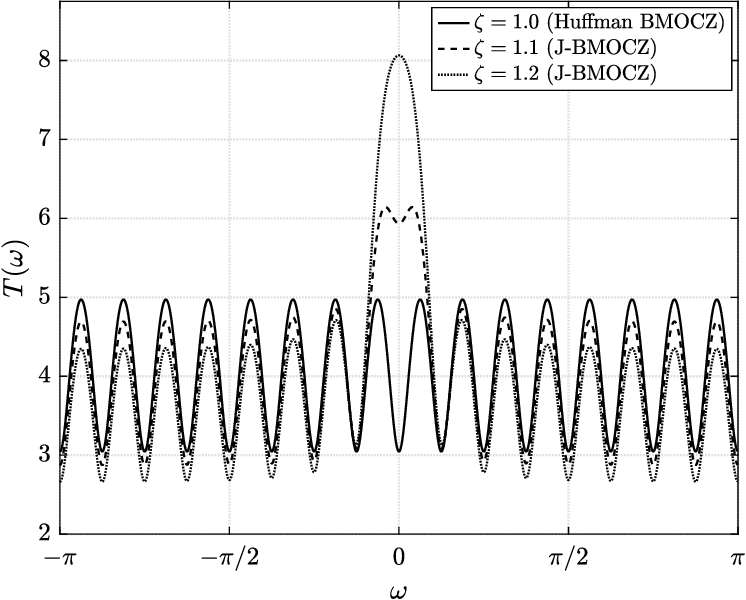}
		\caption{Example template transforms for $\numZeros=16$ and $\radius=1.093$, given by the square root of~\eqref{eq:jutted_psd}.}
		\label{fig:templates}
	\end{figure} 
	
	\subsection{Reliability Metric for Measuring Zero Stability} \label{subsec:stability}
	
	Zero stability is of paramount importance for \ac{BMOCZ}. If~the roots of the received polynomial shift significantly under additive noise perturbing the coefficients, then the receiver will struggle to identify the intended zero pattern. Such instability severely degrades performance with the \ac{dizet} decoder, since any confusion of the zeros induces bit errors. Ensuring the stability of each possible zero pattern is therefore essential to guaranteeing reliable communication. However, analyzing zero stability is challenging, as the relationship between a polynomial's coefficients and its roots is highly non-linear. While there exists relevant literature~\cite{farouki2007root,marden1949geometry}, conventional stability metrics, such as that derived from Rouch\'e's theorem in~\cite{walk2019principles}, only provide an upper bound on the root perturbation for a given noise variance. To the best of our knowledge, there is no general metric to measure the zero stability of~an \emph{arbitrary} polynomial, which we address next.
	
	Assume $\X(z)\in\C[z]$ is a polynomial of degree $\numZeros$ with the zeros $\{\zero_0,\zero_1,\hdots,\zero_{\numZeros-1}\}$ and the $\numZeros+1$ coefficients $\txSeq=[x_0,x_1,\hdots,x_\numZeros]\trans\in\C^{\numZeros+1}$, which satisfy $||\txSeq||_2^2=1$.\footnote{The energy of a polynomial's coefficients does not affect its zero locations. Hence, we assume the coefficients are normalized for simplicity.} Additionally, let $\W(z)\in\C[z]$ be a noise polynomial that perturbs $\X(z)$ as $\Y(z)=\X(z)+\W(z)$. We seek to quantify the stability of each zero $\zero_\zeroIndex$, $\zeroIndex\in[\numZeros]$. To that end, consider
	\begin{equation} \label{eq:noisy_channel}
		\Y(z) = (z-\zero_\zeroIndex)\Hspecial_\zeroIndex(z) + \W(z)\,,
	\end{equation} 
	where we define $\Hspecial_\zeroIndex(z)$ as 
	\begin{equation} \label{eq:channel_poly}
		\Hspecial_\zeroIndex(z) \triangleq \frac{\X(z)}{z-\zero_\zeroIndex} = x_\numZeros \prod_{\substack{\kprime=0\\\kprime\neq k}}^{\numZeros-1} (z-\zero_{\kprime})\,, \ \  \zeroIndex\in[\numZeros]\,.
	\end{equation}
	In particular, notice that~\eqref{eq:noisy_channel} resembles the form of~\eqref{eq:rx_poly}, which motivates our following interpretation: $\Hspecial_\zeroIndex(z)$ acts as a \emph{frequency-selective} channel formed by the neighboring zeros of $\zero_\zeroIndex$, i.e., we transmit $(z-\zero_\zeroIndex)$ in~\eqref{eq:noisy_channel}, and the receiver observes $\Y(z)$. With this interpretation, the capacity of $\Hspecial_\zeroIndex(z)$ corresponds to the \emph{reliability} of $\zero_\zeroIndex$. Towards computing this capacity, let us probe the channel $\Hspecial_\zeroIndex(z)$ with a unit input such that~\eqref{eq:noisy_channel} simplifies to $\Ymod(z)=\Hspecial_\zeroIndex(z)+W(z)$. Then, analogous to \ac{OFDM}, we parallelize $\Hspecial_\zeroIndex(z)$ across frequency by evaluating it at the $\idftSize$th roots of unity, which yields the samples $\Ymod(\e^{j2\pi\idftIndex/\idftSize}) = \Hspecial_\zeroIndex(\e^{j2\pi\idftIndex/\idftSize})+W(\e^{j2\pi\idftIndex/\idftSize})$, $\idftIndex\in[\idftSize]$. Although $\W(z)$ may be arbitrary, here we assume $\W(\e^{j2\pi\idftIndex/\idftSize})\sim\complexnormal(0,1)$, independently across $\idftIndex\in[\idftSize]$. With these simplifications, the transformed channel reduces~to a \emph{parallel AWGN channel}, where each sub-channel has unit noise variance and an input of unit power. The capacity of such a channel is given by \cite[Eq.~(5.38)]{tse2005fundamentals}
		\begin{equation} \label{eq:zero_capacity}
			\capacity_\zeroIndex = \sum_{\idftIndex=0}^{\idftSize-1} \log_2 \left( 1 + \left|\Hspecial_\zeroIndex(\e^{j2\pi\frac{\idftIndex}{\idftSize}})\right|^2 \right)\,, \ \ \zeroIndex\in[\numZeros]\,,
		\end{equation}
		which we scale by $1/\idftSize$ as $\estCapacity_\zeroIndex=\capacity_\zeroIndex/\idftSize$ and assign as a metric for the stability of $\zero_\zeroIndex$. To measure the total stability of~$\X(z)$, taking into account all $\numZeros$ of its zeros, we simply average over $\estCapacity_\zeroIndex$ for $\zeroIndex\in[\numZeros]$ as
	\begin{equation} \label{eq:seq_capacity}
		\capacity(\txSeq) = \frac{1}{\numZeros} \sum_{\zeroIndex=0}^{\numZeros-1} \estCapacity_\zeroIndex\,.
	\end{equation}
	We use the notation $\capacity(\txSeq)$ to emphasize that~\eqref{eq:seq_capacity} quantifies the zero stability of the polynomial $\X(z)$, defined by its (normalized) coefficients $\txSeq\in\C^{\numZeros+1}$. This metric naturally extends to the codebook level, where we estimate stability by averaging over $\capacity(\txSeq)$ for the $2^\numZeros$ codewords $\txSeq\in\polyCodebook$ as
	\begin{equation} \label{eq:mean_capacity}
		\meanCapacity(\polyCodebook) = \frac{1}{2^\numZeros} \sum_{\txSeq\in\polyCodebook} \capacity(\txSeq)\,.
	\end{equation}
	Since $|\polyCodebook|$ scales exponentially with $\numZeros$, computing $\meanCapacity(\polyCodebook)$ via~\eqref{eq:mean_capacity} becomes computationally intractable for $\numZeros\gg1$. However, an approximation of $\meanCapacity(\polyCodebook)$ is readily obtained by averaging over a random subset of codewords. 
	
	We now illustrate our proposed reliability metric with several examples. In each example, we set $\idftSize=1024$ and normalize $||\txSeq||_2^2$ to one.
	
	\begin{example} \label{ex:e1}
		The Huffman \ac{BMOCZ} codebook with $\numZeros=8$ and $\radius=1.176$ has a stability $\meanCapacity(\polyCodebook_{\text{H}})\approx1.149$. The most stable of such Huffman polynomials has all zeros of magnitude $\radius^{-1}$ and a reliability $\capacity(\txSeq_{\radius^{-1}})=1.250$, while the least stable has all zeros of magnitude $R$ and a reliability $\capacity(\txSeq_\radius)=1.048$.
	\end{example}
	
	\begin{example} \label{ex:e2}
		Wilkinson's polynomial~\cite{wilkinson1984perfidious}, which is given by $\Xwilk(z)=\prod_{\zeroIndex=1}^{20}(z-\zeroIndex)$, famously exhibits instability under noise. If the coefficient $x_{19}$ is decreased by just $2^{-23}$, then the zero structure of the polynomial is destroyed. We compute $\capacity(\txSeq_\mathrm{W})=0.0381$ with~\eqref{eq:seq_capacity}, a reliability \emph{32 times less} than the least stable Huffman polynomial with $\numZeros=20$ and $\radius=1.075$.   
	\end{example}  
	
	\begin{example} \label{ex:e3}
		In Fig.~\ref{fig:stability_example}, we plot the zeros of a \ac{JBMOCZ} polynomial under \ac{AWGN} for $\numZeros=8$, $\radius=1.176$, $\asymFactor=1.15$. Each zero is shaded according to its reliability $\estCapacity_\zeroIndex$, where a darker shade corresponds to lower reliability  (less stability). Observe that the jutted zero is the least stable and that the reliability of each root correlates with the area of its perturbation ``cloud." The stability of the \ac{JBMOCZ} codebook with these constellation parameters is $\meanCapacity(\polyCodebook_{\text{J}})\approx1.123$.
	\end{example}
	
	\begin{figure}[t!]
		\centering
		\includegraphics[width=3.25in]{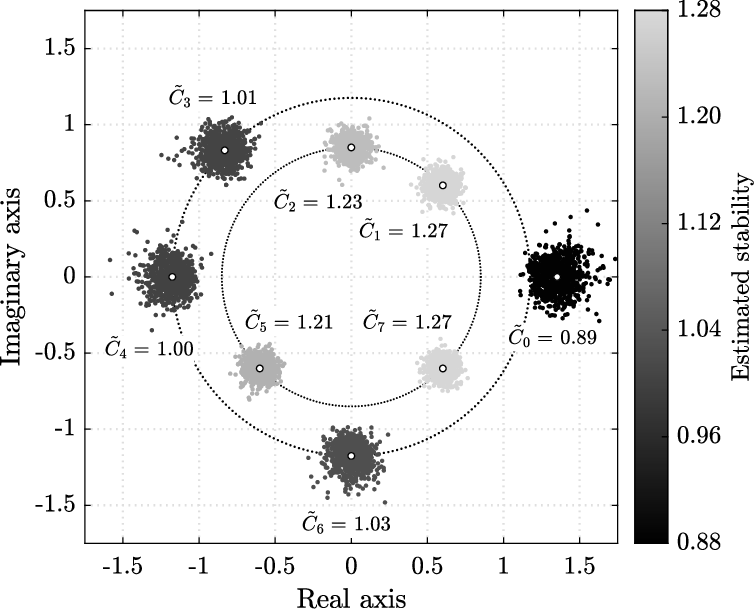}
		\caption{\ac{JBMOCZ} zero perturbation at $\ebno=10$~dB with $\numZeros=8$, $\radius=1.176$, and $\asymFactor=1.15$. The zeros are shaded according to their stability estimated via~\eqref{eq:zero_capacity} as $\estCapacity_\zeroIndex=\capacity_\zeroIndex/\idftSize$, where $\idftSize=1024$.}
		\label{fig:stability_example}
	\end{figure}
	
	\subsection{Optimization of Zero Constellation Parameters} \label{subsec:optimization}
	
	In this section, we attempt to answer the following question: \emph{for a given $\numZeros$, what are the optimal $\radius$ and $\asymFactor$ for \ac{JBMOCZ}?} On one hand, $\asymFactor$ should be small (i.e., $\asymFactor\approx1$) to retain the numerical stability of Huffman \ac{BMOCZ}. On the other hand, if $\asymFactor$ is too small, $\T(\omega)$ does not have a prominent peak, which can lead to erroneous estimates of zero rotation with~\eqref{eq:template_corr}. Hence, there exists a trade-off between zero stability and template ``peakiness" (or, equivalently, rotational asymmetry).
	
	We first introduce a method to choose $\radius>1$ for fixed $\numZeros$ and $\asymFactor$, which is as follows: select the radius $\radiusStar(\numZeros,\asymFactor)$ maximizing the minimum stability in the codebook, i.e.,
	\begin{equation} \label{eq:optimal_radius}
		\radiusStar(\numZeros,\asymFactor) = \arg \max_{\radius>1} \min_{\txSeq\in\polyCodebook} \capacity(\txSeq) \,.
	\end{equation}
	This approach is illustrated in Fig.~\ref{fig:parameter_sweep}(a), which plots the minimum stability $\min_{\txSeq\in\polyCodebook} \capacity(\txSeq)$ against $\radius$ for \ac{JBMOCZ} with $\numZeros=128$ and $\asymFactor\in\{1,1.03,1.06\}$; for reference, we~have annotated $\radiusStar(\numZeros,\asymFactor)$ from~\eqref{eq:optimal_radius}. Notice that for $\asymFactor=1$, $\radiusStar(128,1)=1.015$, which is nearly identical to the conventional radius for Huffman \ac{BMOCZ} with $\numZeros=128$, namely, $\radius=\sqrt{1+\sin(\pi/128)}=1.012$~\cite{walk2019principles,walk2020practical,perre2025learning}. Furthermore, we remark that for both Huffman \ac{BMOCZ} and \ac{JBMOCZ}, the polynomial with all zeros outside (inside) the unit circle is consistently the least (most) stable under~\eqref{eq:seq_capacity}. Of course, prior to calculating $\radiusStar(\numZeros,\asymFactor)$, we must first select $\asymFactor\geq1$ given~$\numZeros$. To~that end, we sweep $\asymFactor$ and record $\min_{\txSeq\in\polyCodebook} \capacity(\txSeq)$ and template \ac{PAPR} [cf.~\eqref{eq:papr_def}]. For each $\asymFactor$, we consider the radius $\radiusStar(\numZeros,\asymFactor)$. Fig.~\ref{fig:parameter_sweep}(b) illustrates this method, where the zero stability is normalized with respect to Huffman \ac{BMOCZ} (i.e., $\asymFactor=1$). The figure clearly depicts the aforesaid trade-off: increasing $\asymFactor$ increases template \ac{PAPR}, but it \emph{strictly~reduces} zero stability. The utility of Fig.~\ref{fig:parameter_sweep} is that it enables one to select $\asymFactor$ for a target template \ac{PAPR} while, at the same time, understanding exactly how much stability is lost relative to Huffman \ac{BMOCZ}. 
	
	\begin{remark} \label{remark:optimization}
		Our presented methods for tuning the parameters of the zero constellation generalize to other \ac{MOCZ} schemes. In fact, given $\numZeros$, one could pose a constrained optimization problem and attempt to identify the zero constellation that maximizes~\eqref{eq:mean_capacity} or $\min_{\txSeq\in\polyCodebook} \capacity(\txSeq)$. However, an optimization of this nature is beyond the scope of this work.  
	\end{remark}
	
	\begin{figure}[t!]
		\centering
		
		\subfloat[Zero stability against radius for $\numZeros=128$.]{\includegraphics[width=3in]{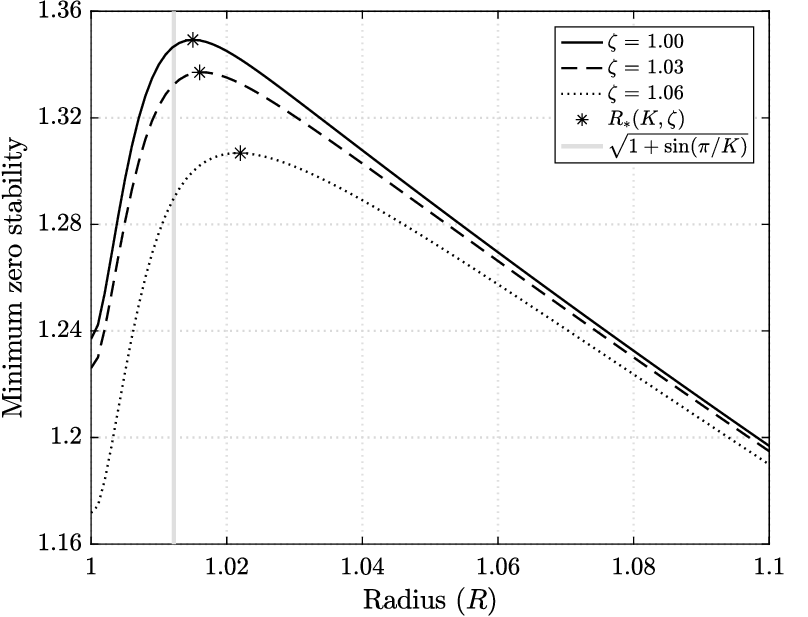}\label{subfig:radius_sweep}}\\
		\subfloat[Zero stability against rotational asymmetry.]{\includegraphics[width=3in]{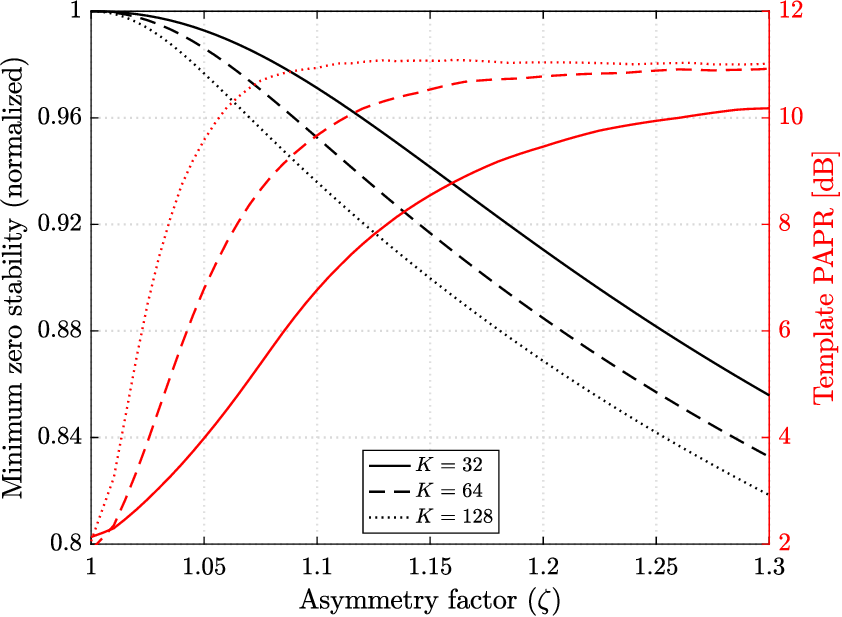}\label{subfig:zeta_sweep}}
		
		\caption{Design curves for \ac{JBMOCZ} parameter selection showing the minimum zero stability identified via~\eqref{eq:seq_capacity} as a function of the radius and asymmetry factor.}	
		\label{fig:parameter_sweep}
	\end{figure}
	
	\section{Pilot-free OFDM Framework} \label{sec:framework}
	
	\begin{figure*}[t!]
		\centering
		\includegraphics[width=6in]{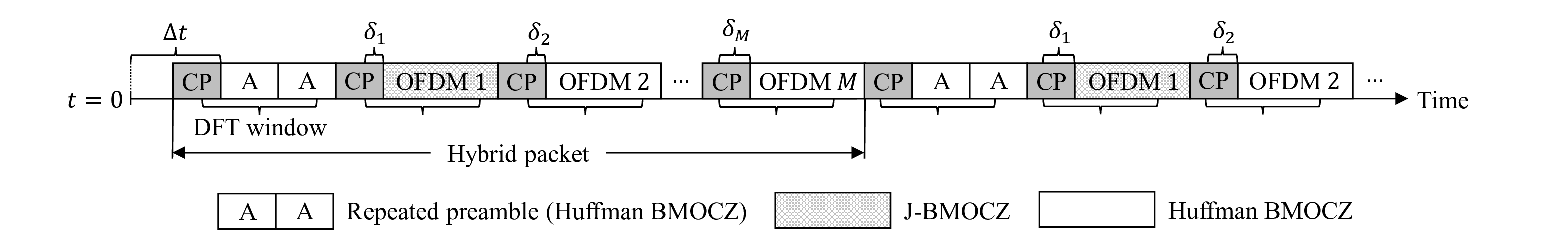}
		\caption{Proposed hybrid \ac{OFDM}-\ac{BMOCZ} waveform with repeated Huffman \ac{BMOCZ} preamble for coarse time synchronization and \ac{CFO} estimation, \ac{JBMOCZ}~symbol for residual \ac{TO} estimation, and Huffman \ac{BMOCZ} payload. Observe that, without correction, the residual \ac{TO} (and hence the \ac{DFT}~window) evolves with the \ac{OFDM} symbol index.}
		\label{fig:ofdm_diagram}
	\end{figure*}
	
	In this section, we combine the concepts introduced in the previous sections to propose an end-to-end, non-coherent \ac{OFDM}-\ac{BMOCZ} framework. Our approach uses a two-stage timing procedure~\cite{lee1997new,yang2000timing,minn2003robust} to address radio impairments at the physical layer while remaining free of pilot symbols.
	
	\subsection{Hybrid Packet Structure} \label{subsec:hybrid}
	
	Consider a payload of $\numPayloadBits$~bits to be transmitted using an \ac{OFDM}-\ac{BMOCZ} waveform. First, we subdivide the bits into $\numPolys=\lceil\numPayloadBits/\numBits\rceil$~blocks, where $\numBits\leq\numZeros$ denotes the number of bits encoded per polynomial,\footnote{If $\numBits$ does not divide $\numPayloadBits$, then $\numPad=B-(\numPayloadBits\bmod\numBits)$~bits must be padded to the payload.} and we allocate $\numSubcarriers=\numZeros+1$ subcarriers and $\numOFDMsymbols=\numPolys$ \ac{OFDM} symbols for the payload. We utilize \ac{JBMOCZ} for the first \ac{OFDM} symbol to estimate the residual \ac{TO} at the receiver, while the remaining $\numOFDMsymbols-1$ \ac{OFDM} symbols employ Huffman \ac{BMOCZ} for its superior zero stability and low \ac{PAPR}. For coarse synchronization and \ac{CFO} estimation, we adopt the Schmidl-Cox technique~\cite{schmidl1997robust} and~prepend a repeated Huffman \ac{BMOCZ} preamble to the payload that encodes a header of $\lfloor\numZeros/2\rfloor$~bits (uncoded). A waveform with the proposed hybrid packet structure is depicted in Fig.~\ref{fig:ofdm_diagram}. The waveform comprises back-to-back hybrid packets, each having a repeated Huffman \ac{BMOCZ} preamble and a \ac{JBMOCZ} symbol to refresh estimates of the \ac{CFO} and residual \ac{TO}, respectively. In practice, refreshing these estimates is only necessary for long frames and not for short packets, but we illustrate it here for completeness. Further, we consider the addition of a \ac{TM}-based preamble for blind channel estimation in Section~\ref{subsec:chest}.  
	
	\subsection{Coarse Synchronization and CFO Estimation} \label{subsec:synch}
	
	Let $\txSignal_\sync(t)$ denote the synchronization symbol prepended to the payload. We choose the \emph{synchronization polynomial} $\X_\sync(z)$ as a Huffman polynomial with the $\Kprime \triangleq\lfloor \numZeros/2\rfloor$ zeros $\bZero_\text{sync}\in\C^{\Kprime}$, which can encode a header of $\Kprime$~bits. The $\kprime$th coefficient of $\X_\sync(z)$, denoted by $x_{\kprime,\sync}$, defines the data symbol $\dataSymbol_{\subcarrierIndex,\sync}$ in~\eqref{eq:mth_tx_symbol} as 
	\begin{equation} \label{eq:sync_symbols}
		d_{\subcarrierIndex,\sync} = 
		\begin{cases}
			x_{\subcarrierIndex/2,\sync}, & \subcarrierIndex \ \text{even}\\
			0, & \text{otherwise}
		\end{cases}\,, \ \ \subcarrierIndex \in [\numSubcarriers]\,.
	\end{equation}
	The coefficients of $X_\sync(z)$ are mapped to every other subcarrier so that, after the \ac{IDFT} operation in~\eqref{eq:mth_tx_symbol}, $\txSignal_\sync(t)$ exhibits a half-symbol repetition, i.e., $\txSignal_\sync(t)=\txSignal_\sync(t+\symbolDuration/2)$ for $t\in[-\cpDuration,\symbolDuration/2)$.
	
	At the receiver, we define $\Nprime \triangleq \idftSize/2$ and search for the repeated preamble using the normalized correlation metric~\cite{schmidl1997robust} $\normCorr_\generalIndex = \sumCorr_\generalIndex/\energyCorr_\generalIndex$, where $\sumCorr_\generalIndex$ accumulates correlated samples as
	\begin{equation} \label{eq:sum_corr}
		\sumCorr_\generalIndex = \sum_{\nprime=0}^{\Nprime-1} \left \{ \rxSignal([\generalIndex+\nprime]/\sampleRate) \times \overline{\rxSignal([\generalIndex+\nprime+\Nprime]/\sampleRate)} \right \}\,,
	\end{equation}
	and $\energyCorr_\generalIndex$ accumulates symbol energy as
	\begin{equation} \label{eq:energy_corr}
		\energyCorr_\generalIndex = \sum_{\nprime=0}^{\Nprime-1} \left | \rxSignal([\generalIndex+\nprime+\Nprime]/\sampleRate) \right |^2\,.
	\end{equation}
	We then obtain an estimate of the \ac{TO} (i.e., $\timingOffset$) by identifying the largest candidate offset $\generalIndex\in\N$ (in samples) for which $\normCorr_\generalIndex$ is within a set fraction $\percentage\in(0,1]$ of $\normCorr_\text{max} \triangleq \max_\generalIndex \normCorr_\generalIndex$ as 
	\begin{equation}
		\timingOffsetEst = \frac{\syncSample}{\sampleRate} \ \ \text{with} \ \ \syncSample = \max \{\generalIndex\in\N \ | \ \normCorr_\generalIndex \geq \percentage  \normCorr_\text{max}\}\,.
	\end{equation}
	Additionally, since the samples correlated in~\eqref{eq:sum_corr} vary by a phase term proportional to the \ac{CFO}~\cite{schmidl1997robust,moose1994technique}, we use  $\sumCorr_{\syncSample}$ to estimate $\cfo$ as
	\begin{equation}
		\cfoEst = -\frac{\Arg(\sumCorr_{\syncSample})}{2\pi\Nprime} \sampleRate\,,
	\end{equation}
	where we assume that $\cfo<\subcarrierSpacing$. After correcting the \ac{CFO} and the residual \ac{TO} $\residualOffset_0$ (see Section~\ref{subsec:to_correction}), the header data can be recovered from the received synchronization polynomial $\Y_\sync(z) = \sum_{\kprime=0}^{\Kprime}y_{\kprime,\sync}z^{\kprime}$ via the \ac{dizet} decoder.
	
	\subsection{Implementation of Residual TO Estimation} \label{subsec:to_correction}
	
	For residual \ac{TO} estimation, we implement~\eqref{eq:template_corr} using the received samples of the first \ac{OFDM} symbol in the payload. Hence, we set $\X_0(z)$ as a \ac{JBMOCZ} polynomial such that the receiver observes $\rxTemplate\in\C^\idftSize$, which comprises the noisy samples $\tilde{v}_\idftIndex = |\rxSignalMod_0(\idftIndex/\sampleRate)| = |\Y_0(\e^{j(2\pi\idftIndex/\idftSize-\po_0)})|$, $\idftIndex\in[\idftSize]$. To~approximate the continuous search over $\intVariable\in[0,2\pi)$ in~\eqref{eq:template_corr}, we quantize the interval $[0,2\pi)$ into $\idftSize$ uniform bins with the Toeplitz matrix $\templateMat\in\R^{N \times N}$ defined as 
	\begin{equation}
		\templateMat \triangleq 
		\begin{bmatrix}
			\T(\omHat_0) & \T(\omHat_{\idftSize-1}) & \cdots & \T(\omHat_1)\\
			\T(\omHat_1) & \T(\omHat_0) & \cdots & \T(\omHat_2)\\
			\vdots & \vdots & \ddots & \vdots\\
			\T(\omHat_{\idftSize-1}) & \T(\omHat_{\idftSize-2}) & \cdots & \T(\omHat_0)
		\end{bmatrix}\,,
	\end{equation} 
	where $\omHat_\idftIndex = 2\pi\idftIndex/\idftSize$ for $n\in[N]$. The residual \ac{TO} (i.e., $\residualOffset_0$) is~then estimated from the inner product (correlation) of $\rxTemplate$ with the columns of $\templateMat$ as 
	\begin{equation} \label{eq:corr_implementation}
		\residualOffsetEst_0 = \frac{\idftIndexHat}{\sampleRate} \ \ \text{with} \ \ \idftIndexHat = \arg \max_{\idftIndex\in[\idftSize]} (\rxTemplate\trans \templateMat)_\idftIndex\,.
	\end{equation}
	In~\eqref{eq:corr_implementation}, the resolution of $\residualOffsetEst_0$ is limited by the \ac{IDFT} size $\idftSize$. However, as we show in Section~\ref{subsec:comparison}, the performance loss for small $\idftSize$ is not large, especially at low $\ebno$ (see Fig.~\ref{fig:rotation_mse}).  
	
	\subsection{PAPR Analysis with FM} \label{subsec:papr}
	
	For \ac{OFDM}-\ac{BMOCZ} with \ac{FM}, recall that an \ac{OFDM} symbol $s(t)$ is given by the corresponding polynomial $\X(z)$ evaluated on the unit circle, i.e., $s(t) = \X(\e^{j2\pi\subcarrierSpacing t})$. In particular, $\A(\e^{j2\pi\subcarrierSpacing t}) = |\X(\e^{j2\pi\subcarrierSpacing t})|^2 = |s(t)|^2$, and hence the baseband signal power at time $t\in[0,\symbolDuration)$ can be expressed directly in terms of the \ac{AACF} $\A(z)$. Indeed, $\subcarrierSpacing=1/\symbolDuration$ implies $\{\e^{j2\pi\subcarrierSpacing t} \ | \ t\in[0,\symbolDuration)\} = \{\e^{j\omega} \ | \ \omega\in[0,2\pi)\}$, so $\A(\e^{j\omega})$ completely characterizes the instantaneous power at baseband. Here we exploit this fact to derive expressions for the \ac{PAPR} of \ac{OFDM}-\ac{BMOCZ} symbols, defined for a given polynomial $\X(z)$ of degree $\numZeros\geq2$ as
	\begin{equation} \label{eq:papr_def}
		\PAPR\{\X(\e^{j\omega})\} \triangleq \frac{\max_{\omega\in[0,2\pi)}|\X(\e^{j\omega})|^2}{\E[|\X(\e^{j\omega})|^2]}\,.
	\end{equation} 
	Furthermore, since $\A(z)=|\X(z)|^2$ is fixed for \ac{BMOCZ} and the coefficients are normalized in energy such that $\E[|\X(\e^{j\omega})|^2]=\numZeros+1$, the \ac{PAPR} with \ac{FM} is \emph{independent} of the message. This contrasts with \ac{TM}, where the \ac{PAPR} is a function of the data, much like conventional \ac{OFDM} systems. We refer the reader to~\cite[Section~IV-C]{huggins2024optimal} for a brief analysis on the \ac{PAPR} with \ac{TM}.
	
	We begin with Huffman \ac{BMOCZ}, where from~\eqref{eq:huffman_auto} we get $\Ahuffman(\e^{j\omHat})=(\numZeros+1)[1-2\etaH\cos(\numZeros\omHat)]$. Since $\cos(\numZeros\omega)$ is periodic and achieves local minima at odd multiples of $\pi/\numZeros$, it follows that $\max_{\omega\in[0,2\pi)}\Ahuffman(\e^{j\omega})=(\numZeros+1)(1+2\etaH)$ and
	\begin{equation} \label{eq:huffman_papr}
		\PAPR\{\Xhuffman(\e^{j\omega})\} = 1 + 2\etaH\,.
	\end{equation}
	In fact, for $\numZeros\geq2$ and $\radius>1$, it holds that $\etaH\in(0,1/2)$, which implies $\PAPR\{\Xhuffman(\e^{j\omega})\}<2$, i.e., the \ac{PAPR} with Huffman \ac{BMOCZ} is bounded above by $10\log_{10}(2)\approx3$~dB. We now consider \ac{JBMOCZ} with $\Ajutted(\e^{j\omega})$ given in~\eqref{eq:jutted_psd}. Here the analysis is less straightforward, as $\Ajutted(\e^{j\omega})$ is aperiodic on $\omega\in[0,2\pi)$, and the existence of a global maximum at $\omega=0$ depends jointly on the parameters $\numZeros$, $\radius$, and $\asymFactor$ (see Fig.~\ref{fig:templates}). Hence, we utilize the below proposition.
	\begin{proposition} \label{prop:max_condition}
		Let $\numZeros\geq2$ and $\radius,\asymFactor>1$. A sufficient condition for $\argmax_{\omega\in[0,2\pi)}\Ajutted(\e^{j\omega})=0$ is
		\begin{equation} \label{eq:sufficient_condition}
			\numZeros^2 < \frac{1}{\etaH}\frac{(a-b)(1-2\etaH)}{[a-2\cos(\pi/\numZeros)][b-2\cos(\pi/\numZeros)]}\,,
		\end{equation}
		where $a\triangleq\asymFactor\radius+\asymFactor^{-1}\radius^{-1}$ and $b\triangleq\radius+\radius^{-1}$.
	\end{proposition}
	\noindent
	The proof is given in Appendix~\ref{app:prop_proof}. For a triple $(\numZeros,\radius,\asymFactor)$ satisfying~\eqref{eq:sufficient_condition}, the \ac{PAPR} with \ac{JBMOCZ} follows readily as 
	\begin{equation} \label{eq:template_papr}
		\PAPR\{\Xjutted(\e^{j\omega})\} = \frac{\etaJ}{\etaH}\frac{a-2}{b-2}(1-2\etaH)\,,
	\end{equation}
	since $2<b<a$ implies $(2-a)/(2-b)=(a-2)/(b-2)$. For reference, we plot the \ac{PAPR} with \ac{JBMOCZ} in Fig.~\ref{fig:papr_space} as a function of $\numZeros$ and $\radius$ for various $\zeta$. Fig.~\ref{fig:papr_space}(a) shows the low \ac{PAPR} afforded by Huffman \ac{BMOCZ} (i.e., $\asymFactor=1$), while Fig.~\ref{fig:papr_space}(b) and Fig.~\ref{fig:papr_space}(c) show that for large $\numZeros$ and small $\radius$, the \ac{PAPR} with \ac{JBMOCZ} can exceed $10$~dB, especially for increasing $\asymFactor$. In Section~\ref{sec:results}, however, we demonstrate good performance with a template \ac{PAPR} less than $9$~dB. 
	
	\begin{remark} \label{remark:angle_coding}
		For fixed $\numZeros$, the \ac{BMOCZ} codebook $\polyCodebook$ consists of the $2^\numZeros$ non-trivial ambiguities $\txSeq\in\C^{\numZeros+1}$ sharing the common \ac{AACF} $\A(z)$~\cite{walk2019principles}. In much the same way, with~\ac{FM}, the \ac{OFDM}-\ac{BMOCZ} signal set comprises the $2^\numZeros$ non-trivially distinct signals $s(t)$ [cf.~\eqref{eq:mth_tx_symbol}] with common envelope $|s(t)|$. {\reviewColor From a time-domain perspective, then, \ac{OFDM}-\ac{BMOCZ} is an implementation of \emph{angle coding}, which was first described by Voelcker in~\cite{voelckertowardp1,voelckertowardp2}. With angle coding, the transmitted waveform has a predetermined envelope, and information is conveyed entirely through its phase trajectory.}
	\end{remark}
	
	\begin{figure*}[t!]
		\centering
		
		\subfloat[$\asymFactor=1.00$ (Huffman \ac{BMOCZ}).]{\includegraphics[width=2.35in]{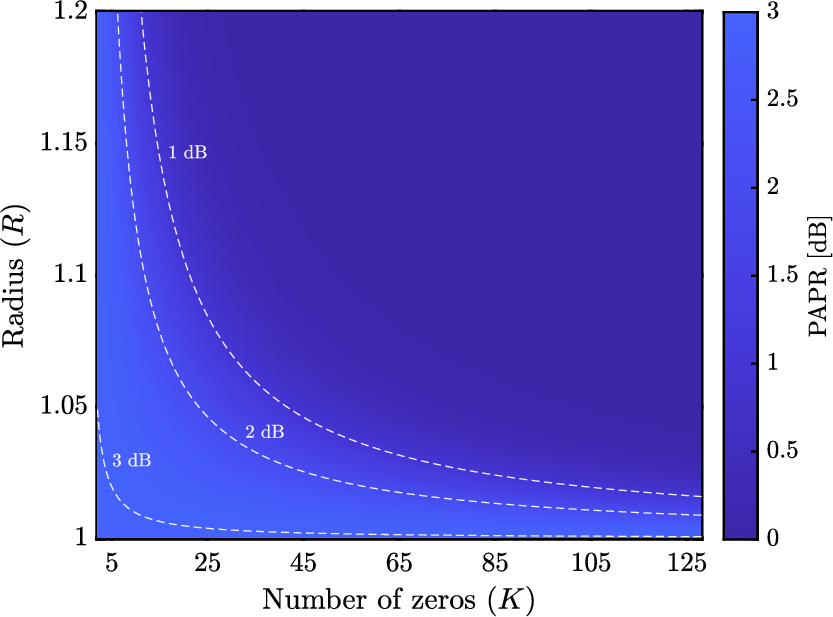}\label{subfig:papr_space_z1}}~\hspace{-7.25pt}
		\subfloat[$\asymFactor=1.05$ (\ac{JBMOCZ}).]{\includegraphics[width=2.35in]{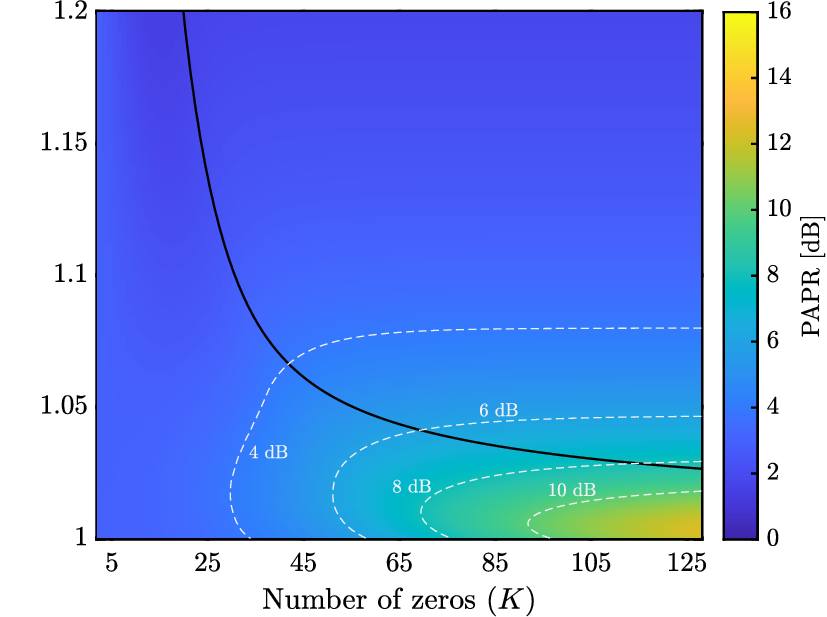}\label{subfig:papr_space_z2}}~\hspace{-7.25pt}
		\subfloat[$\asymFactor=1.10$ (\ac{JBMOCZ}).]{\includegraphics[width=2.35in]{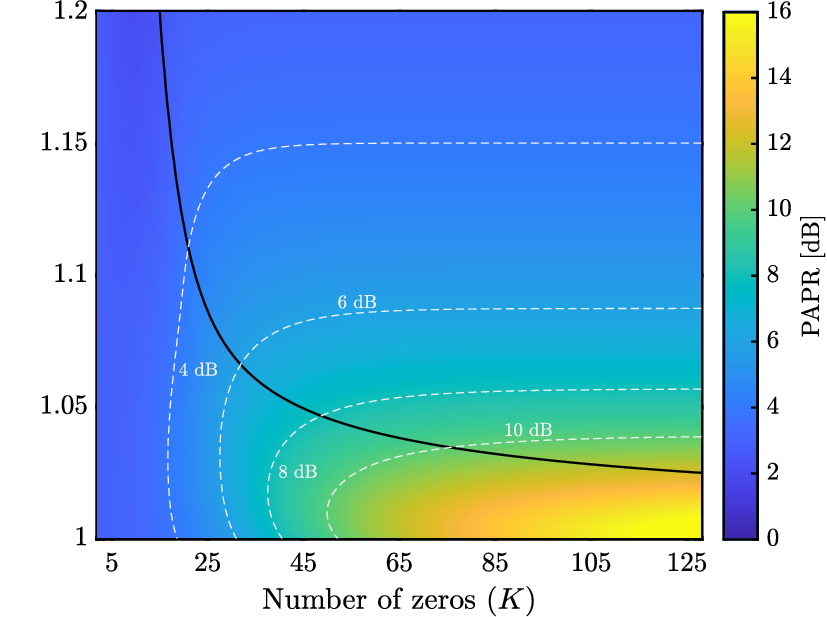}\label{subfig:papr_space_z3}}
		
		\caption{\ac{PAPR} for \ac{OFDM}-\ac{BMOCZ} with \ac{FM}. For a given triple $(\numZeros,\radius,\asymFactor)$, the \ac{PAPR} is fixed and independent of the message. Pairs $(\numZeros,\radius)$ above the black curves in (b) and (c) satisfy the condition in~\eqref{eq:sufficient_condition} for fixed $\asymFactor$.}	
		\label{fig:papr_space}
	\end{figure*}
	
	\subsection{Blind Channel Estimation with TM} \label{subsec:chest}
	
	To improve performance with \ac{FM} under frequency-selective fading, we consider the addition of a \ac{TM}-based preamble for blind channel estimation, placed just after the synchronization symbol in Fig.~\ref{fig:ofdm_diagram}. The preamble comprises $\numSubcarriers=\numZeros+1$ polynomials of degree $\numZerosTM\geq2$ and encodes $\numZerosTM(\numZeros+1)$~bits (uncoded). Let $\X_\subcarrierIndex(z)=\sum_{\ofdmSymbolIndex=0}^{\numZerosTM}x_{\ofdmSymbolIndex,\subcarrierIndex}z^\ofdmSymbolIndex$ be the polynomial transmitted over $\numZerosTM+1$ consecutive \ac{OFDM} symbols of the subcarrier $\subcarrierIndex\in[\numSubcarriers]$. At the receiver, after \ac{CP} removal and a \ac{DFT} operation, we observe the polynomial
	\begin{equation}
		\Y_\subcarrierIndex(z) = \sum_{\ofdmSymbolIndex=0}^{\numZerosTM} x_{\ofdmSymbolIndex,\subcarrierIndex}\subcarrierGain_{\subcarrierIndex,\ofdmSymbolIndex}z^\ofdmSymbolIndex + W(z) \approx \subcarrierGain_\subcarrierIndex\X_\subcarrierIndex(z) + W(z)\,,
	\end{equation}
	where the approximation follows from our quasi-static channel assumption and the prior \ac{CFO} compensation of Section~\ref{subsec:synch}.\footnote{We note that for \ac{TM}, it is the \ac{CFO} that induces zero rotation, not a \ac{TO}. In fact, a \ac{TO} has no effect on the zeros with \ac{TM}.} Decoding $\Y_\subcarrierIndex(z)$ via~\eqref{eq:hd_dizet} gives $\hat{\bMessage}_\subcarrierIndex\in\binaryField^{\numZerosTM}$, which we re-encode to obtain an estimate of the transmitted polynomial $\X_\subcarrierIndex(z)$ as $\Xhat_\subcarrierIndex(z)=\sum_{\ofdmSymbolIndex=0}^{\numZerosTM}\hat{x}_{\ofdmSymbolIndex,\subcarrierIndex}z^\ofdmSymbolIndex$. With $\Y_\subcarrierIndex(z)$ and $\Xhat_\subcarrierIndex(z)$, we compute the channel estimate $\channelEst_\subcarrierIndex\in\C$ by solving the least-squares problem $\channelEst_\subcarrierIndex=\arg\min_{\subcarrierGain_\subcarrierIndex\in\C}\sum_{\ofdmSymbolIndex=0}^{\numZerosTM}|y_{\ofdmSymbolIndex,\subcarrierIndex}-\hat{x}_{\ofdmSymbolIndex,\subcarrierIndex}\subcarrierGain_\subcarrierIndex|^2$, whose solution for $y_{\ofdmSymbolIndex,\subcarrierIndex} \triangleq x_{\ofdmSymbolIndex,\subcarrierIndex}\subcarrierGain_{\subcarrierIndex,\ofdmSymbolIndex}$ has the closed form
	\begin{equation} \label{eq:least_squares_chest}
		\channelEst_\subcarrierIndex = \frac{\sum_{\ofdmSymbolIndex=0}^{\numZerosTM}\overline{\hat{x}_{\ofdmSymbolIndex,\subcarrierIndex}}\,y_{\ofdmSymbolIndex,\subcarrierIndex}}{\sum_{\ofdmSymbolIndex=0}^{\numZerosTM}|\hat{x}_{\ofdmSymbolIndex,\subcarrierIndex}|^2}\,, \ \ \subcarrierIndex\in[\numSubcarriers]\,.
	\end{equation}
	Using the channel estimate $\channelEst_\subcarrierIndex$ and the noise variance $N_0$, estimated from the null subcarriers in~\eqref{eq:sync_symbols}, we synthesize the \ac{MMSE} equalizer
	\begin{equation}
		\equalizer_\subcarrierIndex = \frac{\overline{\channelEst_\subcarrierIndex}}{|\channelEst_\subcarrierIndex|^2+N_0/\E[|\hat{x}_{\ofdmSymbolIndex,\subcarrierIndex}|^2]}\,, \ \ \subcarrierIndex\in[\numSubcarriers]\,,
	\end{equation}
	where $\E[|\hat{x}_{\ofdmSymbolIndex,\subcarrierIndex}|^2] = ||\hat{\txSeq}_\subcarrierIndex||_2^2/(\numZerosTM+1) = 1$. Then, prior to decoding the \ac{FM}-based payload, we equalize the coefficients of $\Ymod_\ofdmSymbolIndex(z)$ in~\eqref{eq:zero_rotation} to obtain
	\begin{equation} \label{eq:equalized_poly}
		\hat{\Y}_\ofdmSymbolIndex(z) = \sum_{\subcarrierIndex=0}^{\numZeros}y_{\subcarrierIndex,\ofdmSymbolIndex}\e^{-j\po_\ofdmSymbolIndex\subcarrierIndex}\equalizer_\subcarrierIndex z^\subcarrierIndex \approx \sum_{\subcarrierIndex=0}^{\numZeros}x_{\subcarrierIndex,\ofdmSymbolIndex}z^\subcarrierIndex = \X_\ofdmSymbolIndex(z)\,.
	\end{equation}
	Note that the \ac{TO}-induced phase modulation in~\eqref{eq:equalized_poly} is absorbed by the channel estimates and hence is corrected following equalization. Thus, with a \ac{TM}-based preamble, the use of \ac{JBMOCZ} for residual \ac{TO} estimation is not strictly necessary. Additionally, the reliability of the channel estimates improves as $\numZerosTM$ increases, since the averaging in~\eqref{eq:least_squares_chest} becomes more robust to noise. This improvement, however, comes at the cost of a longer preamble and increased \ac{PAPR}, as the \ac{PAPR} of a Huffman sequence increases with its length~\cite{sasidharan2024alternative}.
	
	\subsection{Soft-decision DiZeT Decoding} \label{subsec:dizet}
	
	As noted in Section~\ref{subsec:problem}, our motivation for \ac{JBMOCZ} is to maintain performance under zero rotation without relying on a \ac{CPC}. This, in turn, increases flexibility and enables decoding using \emph{soft information}. Soft-decision decoding is ideal for \ac{JBMOCZ}, as the asymmetry inherently renders some zeros (and hence bits) more reliable than others (see Fig.~\ref{fig:stability_example}). To~obtain soft information from a received polynomial $\Y(z)$ with the coefficients $\rxSeq\in\C^{\totLength}$, we perform direct zero-testing and compute the \acp{PLLR}~\cite{walk2021multi}
	\begin{equation} \label{eq:LLRs}
		\pseudoLLR_\zeroIndex = \log \left( \frac{\e^{-\zeroMag_\zeroIndex^{-(\totLength-1)}|\Y(\zero_\zeroIndex\raisedOne)|^2}}{\e^{-\zeroMag_\zeroIndex^{\totLength-1}|\Y(\zero_\zeroIndex\raisedZero)|^2}} \right)\,, \ \ \zeroIndex\in[\numZeros]\,,
	\end{equation}
	where we first normalize the coefficient energy as $||\rxSeq||_2^2=1$. We call the LLRs in~\eqref{eq:LLRs} \emph{pseudo-LLRs}, since the distribution of the perturbed zeros is assumed as Gaussian, though the true distribution of the perturbed zeros is unknown. Lastly, we remark that in~\cite{ott2025binary} it is shown coding across multiple polynomials greatly improves the performance. While this approach extends naturally to \ac{FM} with coding performed over consecutive \ac{OFDM} symbols, it is not considered here for the purpose of analyzing the basic limitations of the scheme.
	
	\section{Numerical Results} \label{sec:results}
	
	This section presents the results of numerical simulations and an \ac{SDR} implementation. We first consider baseline simulations without \ac{OFDM} to compare the performance of \ac{JBMOCZ} to Huffman \ac{BMOCZ}. We then perform an \ac{OFDM} simulation under timing errors to assess the hybrid waveform and its extensions. Finally, we demonstrate \ac{OFDM}-\ac{BMOCZ} in practice using low-cost \acp{SDR}.
	
	\subsection{Comparison to Huffman BMOCZ} \label{subsec:comparison}
	
	In this section, we perform sequence-level simulations, as~in \cite{walk2019principles} and~\cite{walk2020practical}, to compare the performance of \ac{JBMOCZ} to Huffman \ac{BMOCZ}, i.e., the state-of-the-art \ac{BMOCZ} scheme. In our simulations, we uniformly sample messages $\bMessage\in\binaryField^\numBits$ and normalize $\txSeq\in\polyCodebook$ as $||\txSeq||_2^2=\numZeros+1$, where $\numZeros\geq\numBits$ represents the codeword length. Each transmitted codeword observes an independent \ac{CIR} realization $\channelIR\in\C^{\effLength}$, and the taps are modeled according to a uniform \ac{PDP}~\cite{walk2019principles,walk2020practical} as $\E[|\channelCoefficient_{\effIndex}|^2]=1/\effLength$, $\effIndex\in[\effLength]$. The channel output is then perturbed by \ac{AWGN} $\noiseSeq\in\C^{\totLength}$ with $\noiseElem_{\totIndex}\sim\complexnormal(0,N_0)$, $\totIndex\in[\totLength]$. Hence, the received coefficients $\rxSeq\in\C^{\totLength}$ are given by $\rxSeq=\txSeq\ast\channelIR+\noiseSeq$. We also consider the \ac{AWGN} channel as a baseline, where $\rxSeq=\txSeq+\noiseSeq\in\C^{\numZeros+1}$. In~simulations with zero rotation, we modulate $\rxSeq$ through an angle $\po$ sampled uniformly from $[0,2\pi)$. Of course, uniformly random zero rotation represents an exaggerated impairment, but it is treated for consistency with~\cite{walk2020practical}. We use the radius $\radius=\sqrt{1+\sin(\pi/\numZeros)}$ proposed in~\cite{walk2019principles} for Huffman \ac{BMOCZ}, while \ac{JBMOCZ} employs the radius $\radiusStar(\numZeros,\asymFactor)$ in~\eqref{eq:optimal_radius} with $\asymFactor>1$ selected for a template \ac{PAPR} of $8.5$~dB. All schemes utilize the hard-decision \ac{dizet} decoder in~\eqref{eq:hd_dizet}, except for polar-coded \ac{JBMOCZ}, which exploits soft information via the \acp{PLLR} in~\eqref{eq:LLRs} and uses a \ac{SIC} decoder. Similar to~\cite{walk2019principles}, we assume knowledge of $\effLength$ at the receiver to ensure proper scaling in~\eqref{eq:hd_dizet} and~\eqref{eq:LLRs}.
	
	\subsubsection{Uncoded Performance} \label{subsubsec:uncoded_performance}
	
	\begin{figure}[t!]
		\centering
		\includegraphics[width=3in]{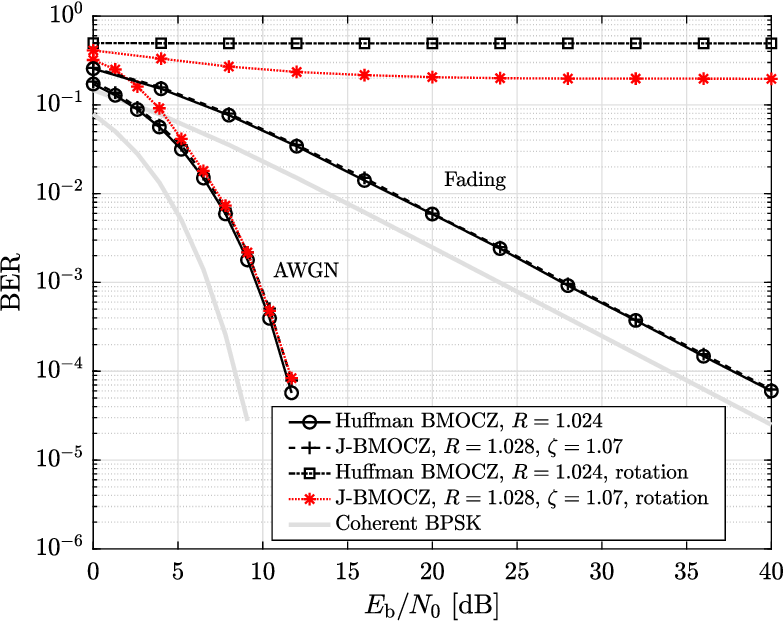}
		\caption{Comparison of uncoded \ac{JBMOCZ} to uncoded Huffman \ac{BMOCZ} for $\numZeros=64$ and $\effLength=5$. The curves marked with ``rotation" are impaired by random zero rotation.}
		\label{fig:error_no_coding}
	\end{figure}
	
	Fig.~\ref{fig:error_no_coding} compares the uncoded \ac{BER} performance of \ac{JBMOCZ} and Huffman \ac{BMOCZ} for $\numZeros=64$ and $\effLength=5$. Without any zero rotation, the \ac{BER} of the two schemes is nearly identical. \ac{JBMOCZ} performs $\leq0.3$~dB worse in both the \ac{AWGN} and fading channels. This result is expected, as with the constellation parameters in Fig.~\ref{fig:error_no_coding}, we compute $\meanCapacity(\polyCodebook_{\text{J}})\approx1.305<\meanCapacity(\polyCodebook_{\text{H}})\approx1.337$; the stability of the \ac{JBMOCZ} codebook is marginally less than that of Huffman \ac{BMOCZ}. More generally, we conjecture that the \ac{BER} of a \ac{BMOCZ} scheme is correlated with the stability in~\eqref{eq:mean_capacity}. Establishing such a relationship, however, requires theoretical \ac{BER} analysis with the \ac{dizet} decoder, which is challenging due to the non-linear relationship between the zeros and the coefficients. {\reviewColor Under zero rotation, notice that Huffman \ac{BMOCZ} experiences an immediate error floor, since without the \ac{ACPC}, integer rotations $\intCFO\baseAngle$ cannot be resolved (see Section~\ref{subsec:problem}).} In contrast, \ac{JBMOCZ} incurs a modest loss in performance under \ac{AWGN} that decreases with $\ebno$. Yet, \ac{JBMOCZ} also experiences an error floor in the selective channel, as the channel zeros distort the received template non-linearly and hinder the correlation in~\eqref{eq:template_corr}. Interestingly, we show next that coding largely resolves this issue.  
	
	\subsubsection{Coded Performance} \label{subsubsec:coded_performance}
	
	\begin{figure}[t!]
		\centering
		
		\subfloat[Coded \ac{BER} performance under zero rotation.]{\includegraphics[width=3in]{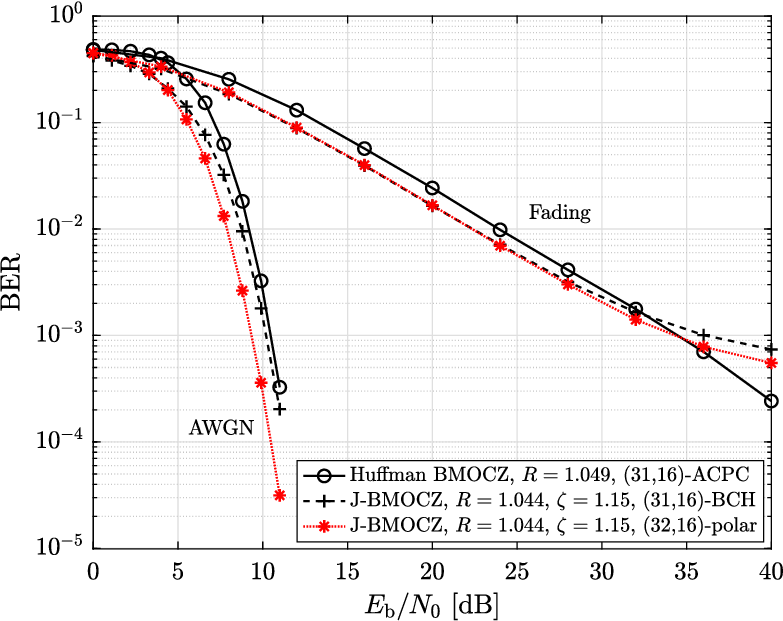}\label{subfig:ber_with_rotation}}\\
		\subfloat[Coded \ac{BLER} performance under zero rotation.]{\includegraphics[width=3in]{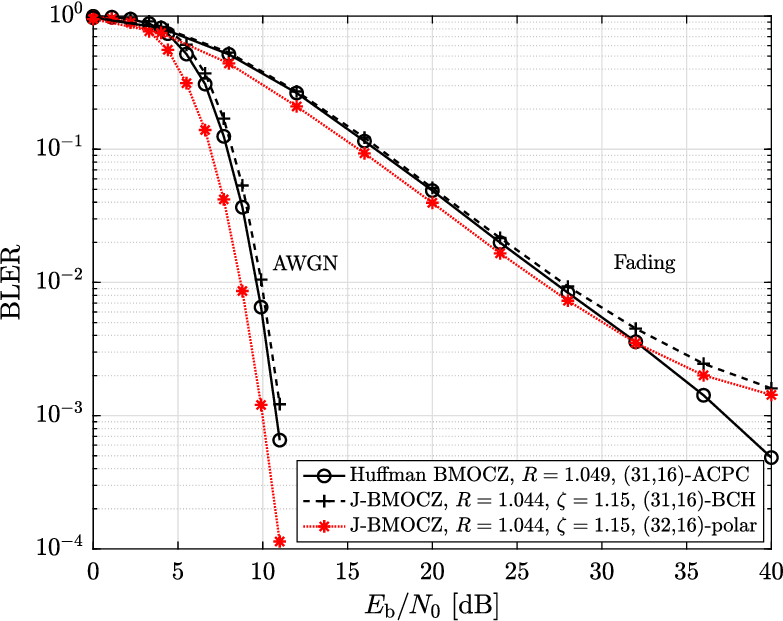}\label{subfig:bler_with_rotation}}
		
		\caption{Comparison of coded \ac{JBMOCZ} to coded Huffman \ac{BMOCZ} under random zero rotation for $\numZeros\in\{31,32\}$ and $\effLength=5$. All schemes encode $\numBits=16$~bits per polynomial.}	
		\label{fig:error_with_coding}
	\end{figure} 
	
	We compare the performance of coded \ac{JBMOCZ} to coded Huffman \ac{BMOCZ} with a message length of $\numBits=16$ and a codeword length of $\numZeros\in\{31,32\}$. Both schemes are impaired by uniformly random zero rotation. For Huffman \ac{BMOCZ}, to correct the fractional rotation $\fracCFO\baseAngle$, we adopt the oversampled \ac{dizet} decoder~\cite[Eq.~(34)]{walk2020practical} with an oversampling factor of $\overFactor=200$. Then, to correct the integer rotation $\intCFO\baseAngle$, we employ the (31,16)-\ac{ACPC}~\cite{walk2020practical}, which inherits a 2-bit error-correction capability from its outer (31,21)-BCH code. For \ac{JBMOCZ}, we consider both a (31,16)-BCH code and a (32,16)-polar code,\footnote{We consider the (31,16)-BCH code with hard-decision decoding for a ``fair" comparison to the (31,16)-ACPC, which also relies on hard decisions. However, the precise advantage of \ac{JBMOCZ} is that it enables soft decoding.} and we estimate $\po$ using the approach of Section~\ref{subsec:to_correction} as $\poEst=2\pi\idftIndexHat/\idftSize$, where $\idftSize=1024$ and $\idftIndexHat\in[\idftSize]$ is defined as in~\eqref{eq:corr_implementation} (see also \cite{huggins2025fourier}). 
	
	Fig.~\ref{fig:error_with_coding} shows the coded \ac{BER} and \ac{BLER} performance under uniformly random zero rotation with $\effLength=5$. In \ac{BER}, both \ac{JBMOCZ}-BCH and polar-coded \ac{JBMOCZ} outperform Huffman \ac{BMOCZ}-\ac{ACPC} at moderate-to-low $\ebno$, with polar-coded \ac{JBMOCZ} achieving gains of $1.5$~dB and $1$~dB in the \ac{AWGN} and fading channels, respectively. In \ac{BLER}, Huffman \ac{BMOCZ}-\ac{ACPC} just outperforms \ac{JBMOCZ}-BCH, but polar-coded \ac{JBMOCZ} again achieves a $1$~dB gain relative to Huffman \ac{BMOCZ}-\ac{ACPC}. {\reviewColor This result highlights the efficacy of soft decoding for \ac{JBMOCZ}.} Note, however, that \ac{JBMOCZ} still experiences an error floor at high $\ebno$ due to the selectivity of the channel. Hence, if the channel is especially selective (e.g., if $\effLength>\numZeros$), then low code rates are required.  
	
	\subsubsection{Estimation Performance} \label{subsubsec:est_performance}
	
	\begin{figure}[t!]
		\centering
		\includegraphics[width=3in]{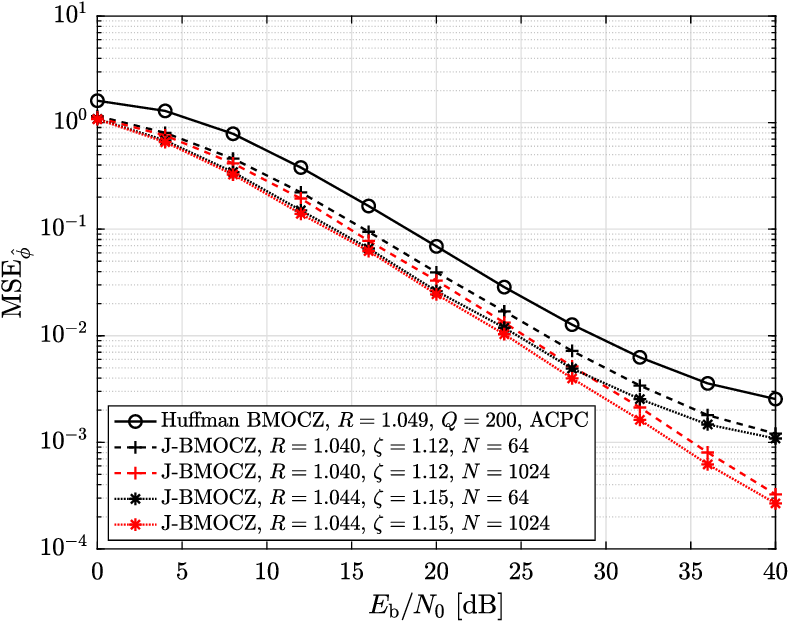}
		\caption{\ac{MSE} of zero rotation estimators for \ac{JBMOCZ} and Huffman \ac{BMOCZ} with $\numZeros=31$ and $\effLength=1$.}
		\label{fig:rotation_mse}
	\end{figure}
	
	We compare the \ac{MSE} of the \ac{JBMOCZ} and Huffman \ac{BMOCZ}-\ac{ACPC} zero rotation estimators for $\numZeros=31$ and $\effLength=1$ (Rayleigh fading). The \ac{MSE} is computed experimentally by averaging over $\numPolys$ estimates as 
	\begin{equation}
		\MSE_{\poEst} = \frac{1}{\numPolys} \sum_{\polyIndex=0}^{\numPolys-1} \min \left\{ [\po_\polyIndex-\poEst_\polyIndex]^2, [\po_\polyIndex-(2\pi-\poEst_\polyIndex)]^2 \right\}\,,
	\end{equation}
	where $\po_\polyIndex$ and $\poEst_\polyIndex$ denote, respectively, the true and estimated zero rotation for the $\polyIndex$th polynomial. For Huffman \ac{BMOCZ}, we utilize the oversampled \ac{dizet} decoder with $\overFactor=200$ and the (31,16)-\ac{ACPC} to obtain $\hat{\fracCFO}$ and $\hat{\intCFO}$, respectively, which combine to yield the composite estimate $\poEst=(\hat{\intCFO}+\hat{\fracCFO})\baseAngle$. For~\ac{JBMOCZ}, we compute $\poEst$ as described in Section~\ref{subsubsec:coded_performance}, and here we consider $\asymFactor\in\{1.12,1.15\}$ and $\idftSize\in\{64,1024\}$. Fig.~\ref{fig:rotation_mse} plots $\MSE_{\poEst}$ against $\ebno$, where the \ac{JBMOCZ} estimators outperform that of Huffman \ac{BMOCZ}-\ac{ACPC} by a margin $\geq2.6$~dB. Notice that increasing $\idftSize$ only marginally improves the \ac{MSE} for \ac{JBMOCZ}, especially at low $\ebno$.
	
	\subsection{OFDM Simulation with Residual TO} \label{subsec:ofdm_sim}
	
	\begin{figure}[t!]
		\centering
		
		\subfloat[\ac{BER} performance under residual \ac{TO}.]{\includegraphics[width=3in]{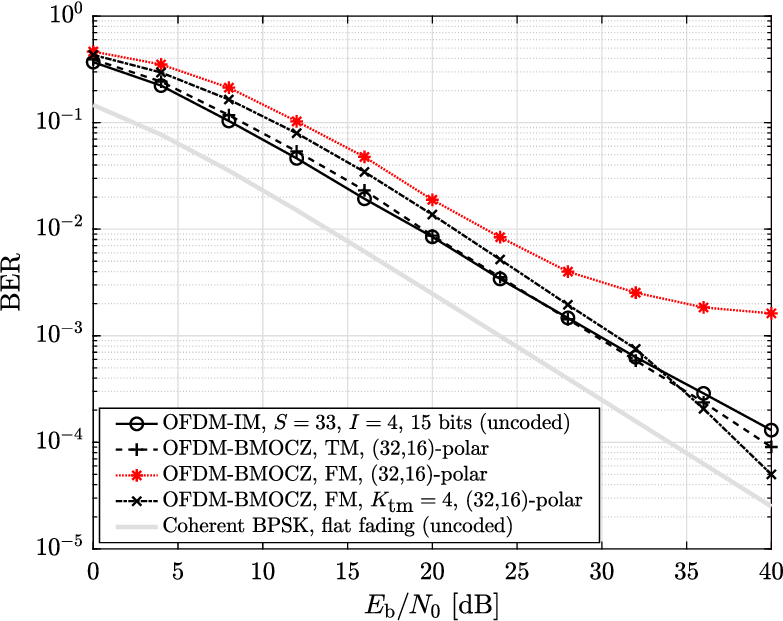}\label{subfig:ber_ofdm}}\\
		\subfloat[\ac{BLER} performance under residual \ac{TO}.]{\includegraphics[width=3in]{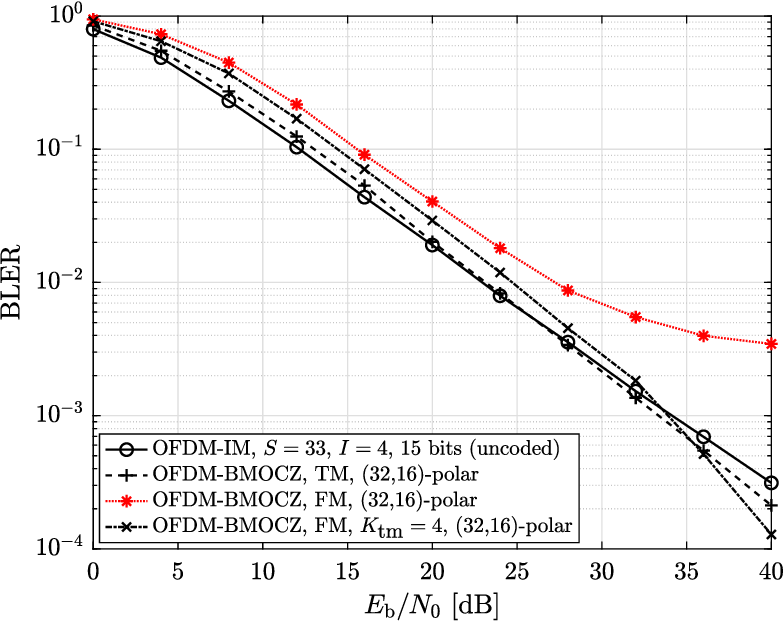}\label{subfig:bler_ofdm}}
		
		\caption{Comparison of \ac{OFDM} schemes in 802.11n Channel-B. \ac{FM} uses the hybrid packet structure of Fig.~\ref{fig:ofdm_diagram}. All schemes are impaired by a residual \ac{TO} of $\stepBackSamples\in[6]$~samples.}	
		\label{fig:error_curves_ofdm}
	\end{figure}
	
	In this section, we perform \ac{OFDM}-based simulations in 802.11n Channel-B~\cite{erceg2004tgn}, which has a delay spread of $80$~ns and models multi-path propagation in an indoor, residential~area. For \ac{OFDM}-\ac{BMOCZ}, we set $\numPayloadBits=512$, $\numZeros=32$, $\numBits=16$, and we compare three different schemes: \ac{FM}, \ac{FM} with blind channel estimation, and \ac{TM}. \ac{FM} uses the hybrid packet of Fig.~\ref{fig:ofdm_diagram} (minus the synchronization preamble) with $\numSubcarriers=\numZeros+1$ subcarriers and $\numOFDMsymbols=\lceil\numPayloadBits/\numBits\rceil=32$ \ac{OFDM} symbols. For blind channel estimation, we select $\numZerosTM=4$ and prepend a \ac{TM}-based, Huffman \ac{BMOCZ} preamble of $\numZeros+1$ subcarriers and $\numZerosTM+1$ \ac{OFDM} symbols to the \ac{FM} waveform. \ac{TM} is also based on Huffman \ac{BMOCZ} with $\numSubcarriers=\lceil\numPayloadBits/\numBits\rceil=32$ subcarriers and $\numOFDMsymbols=\numZeros+1$ \ac{OFDM} symbols. We use the conventional radius $\radius=\sqrt{1+\sin(\pi/\numZeros)}=1.048$~\cite{walk2019principles} for Huffman \ac{BMOCZ} (similarly, $\radius_{\text{tm}}=1.307$), while \ac{JBMOCZ} uses $\asymFactor=1.15$ and $\radiusStar(\numZeros,\asymFactor)=1.044$. A (32,16)-polar code is employed for all \ac{OFDM}-\ac{BMOCZ} schemes. As a baseline, we~simulate uncoded, non-coherent \ac{OFDM}-\ac{IM} with $\numZeros+1$ subcarriers, $\imActiveCarriers=4$ of which are active; this enables the encoding of $\lfloor\log_2\binom{\numZeros+1}{\imActiveCarriers}\rfloor=15\approx\numBits$~bits per \ac{OFDM} symbol. We fix $\sampleRate=10$~MHz, $\idftSize=256$, and $\cpSize=8$, which yields an approximate bandwidth of $1.3$~MHz for all schemes. To~simulate a residual \ac{TO}, we assume perfect synchronization and then randomly back off $\stepBackSamples\in[6]$~samples. For \ac{FM}, both with and without blind channel estimation, we correct the residual \ac{TO} using~\eqref{eq:corr_implementation}. For \ac{TM} and \ac{OFDM}-\ac{IM}, however, a timing error has no effect and so no correction is considered. A summary of the simulated schemes is provided in Table~\ref{tab:sim_summary}.
	
	Fig.~\ref{fig:error_curves_ofdm} shows the simulated \ac{BER} and \ac{BLER} performance against $\ebno$. Measured at a \ac{BER} of $10^{-1}$, \ac{FM} performs $3$~dB and $4$~dB worse than \ac{TM} and \ac{OFDM}-\ac{IM}, respectively. Adding the \ac{TM}-based preamble for blind channel estimation reduces the gap to $2$~dB and $3$~dB, respectively. In fact, \ac{FM} with blind channel estimation outperforms all schemes at high $\ebno$ due to the reliability of the channel estimates in the absence of noise. This result suggests that further averaging at low $\ebno$ to denoise the channel estimates could significantly improve the performance. As noted in Section~\ref{subsec:chest}, however, averaging requires a longer preamble, which is undesirable for small payloads. Still, we remark that other techniques can improve the performance with \ac{FM}. Multiple receive antennas, for example, would facilitate frequency diversity and enhance both detection~\cite{walk2020practical,sun2023non} and the reliability of the blind channel estimates. Time-frequency coding represents another natural extension~\cite{liu2002space}. 
	
	\begin{table}[t!]
		\centering 
		\caption{Summary of simulated OFDM schemes}
		\label{tab:sim_summary}
		\begin{tabular}{lcccc}
			\toprule
			& FM & FM + CHEST & TM & IM\\
			\midrule
			Subcarriers & 33 & 33 & 32 & 33\\
			OFDM symbols & 32 & 38 & 33 & 34\\
			Spectral eff. [bits/s/Hz] & 0.47 & 0.49 & 0.46 & 0.44\\
			\bottomrule
		\end{tabular}
	\end{table}
	
	\begin{figure*}[t!]
		\centering
		
		\subfloat[Experimental setup.]{\includegraphics[width=2in]{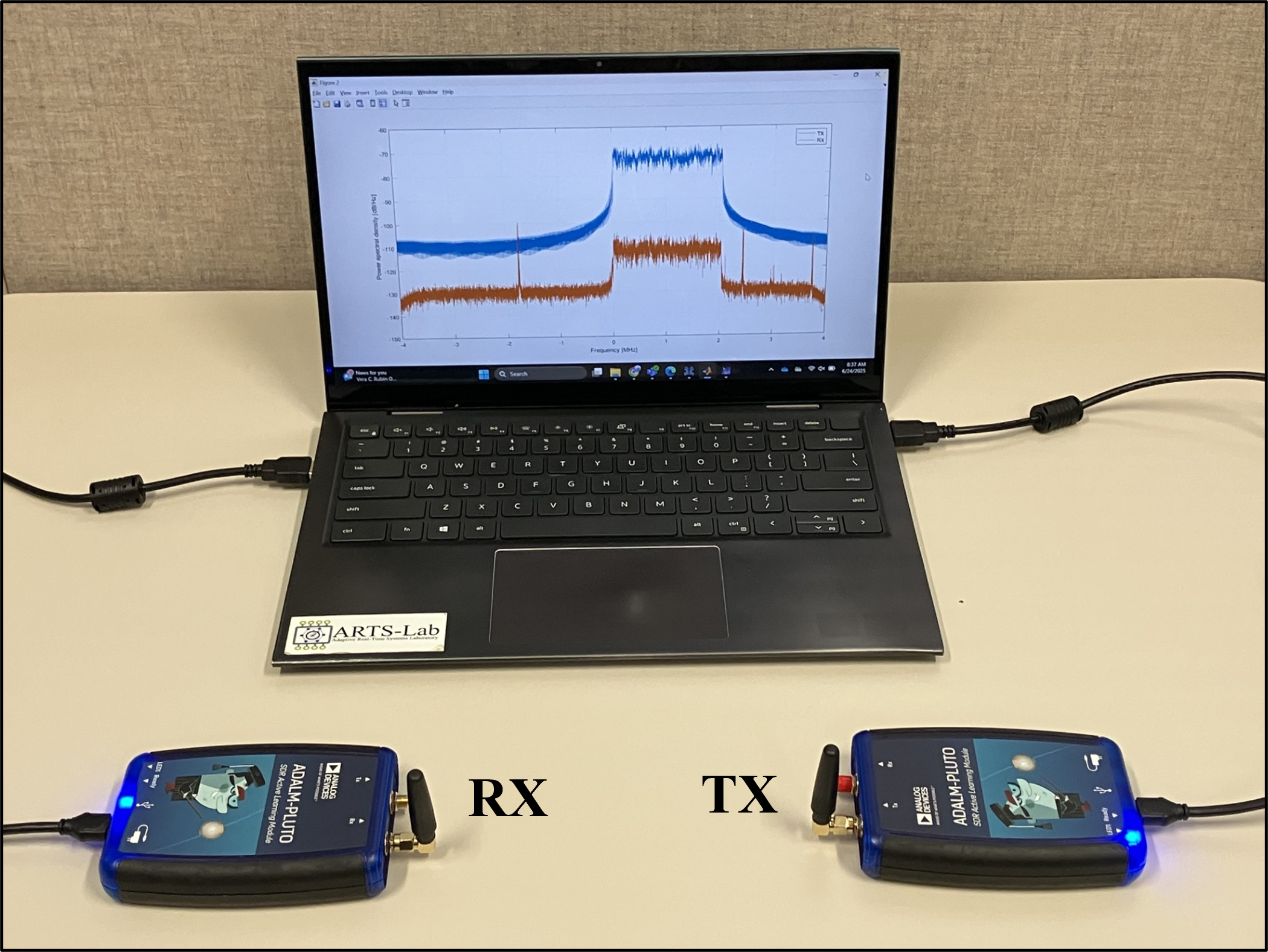}\label{subfig:demo_setup}}~\hspace{0.5cm}
		\subfloat[I \& Q components of the TX signal.]{\includegraphics[width=2in]{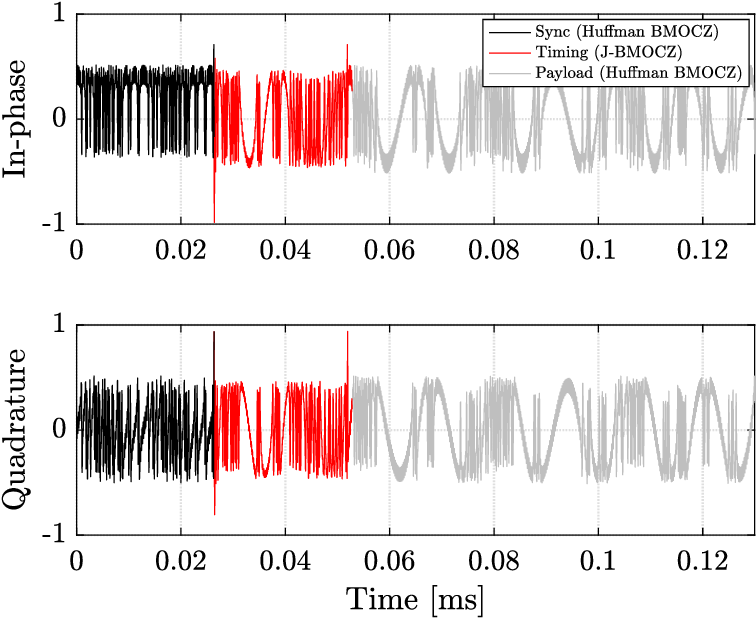}\label{subfig:tx_data}}~\hspace{0.5cm}
		\subfloat[PSD of the TX and RX signals.]{\includegraphics[width=2in]{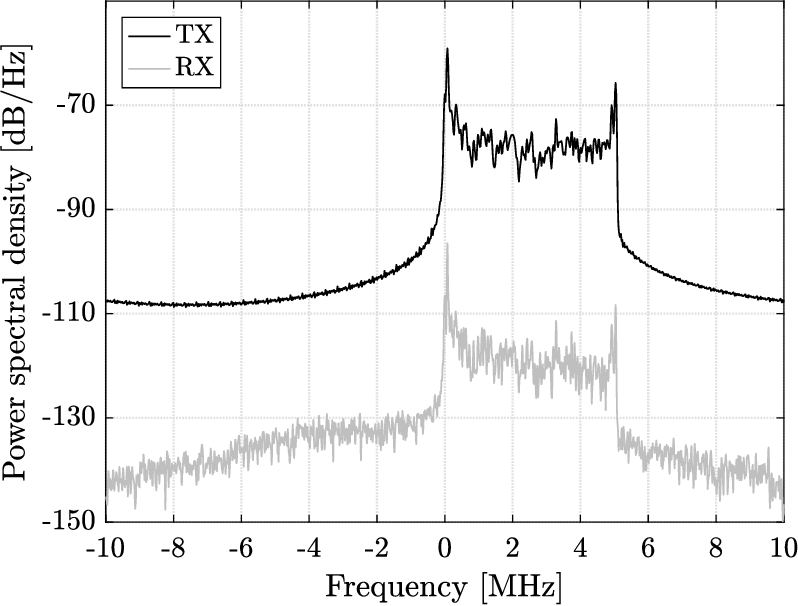}\label{subfig:psd}}
		
		\subfloat[I \& Q components of the RX signal.]{\includegraphics[width=2.1in]{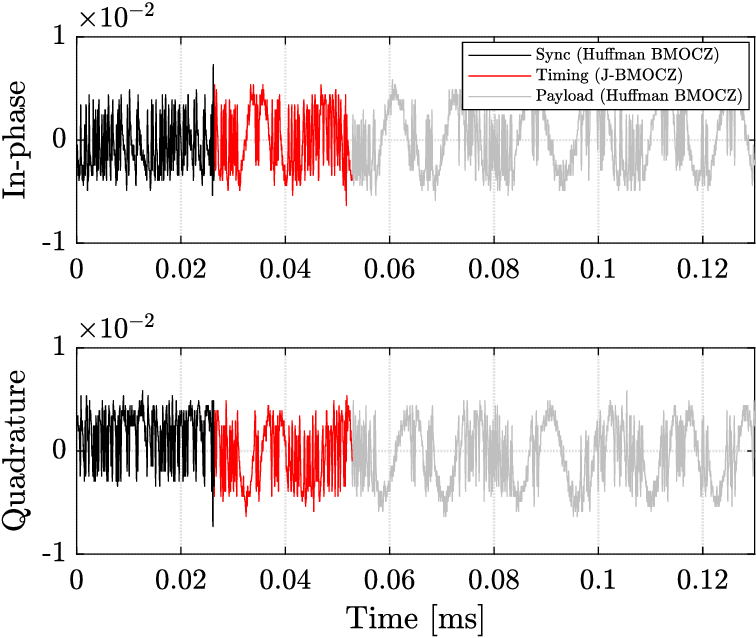}\label{subfig:rx_data}}~\hspace{0.5cm}
		\subfloat[Normalized template transforms.]{\includegraphics[width=2in]{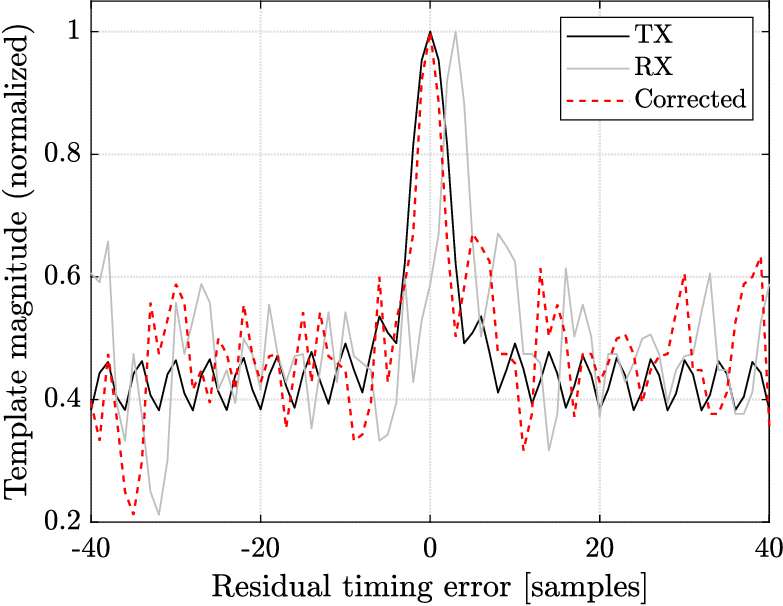}\label{subfig:tx_rx_templates}}~\hspace{0.6cm}
		\subfloat[Zeros of the \ac{JBMOCZ} polynomial.]{\includegraphics[width=1.8in]{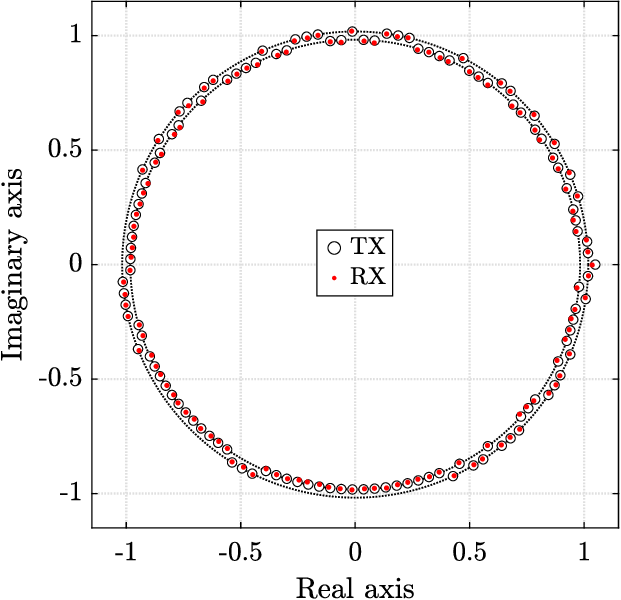}\label{subfig:timing_zeros}}

		\caption{\ac{SDR} demonstration of the proposed hybrid waveform in Fig.~\ref{fig:ofdm_diagram}. (a) Experimental setup with two ADALM-PLUTO \acp{SDR}. (b) TX signal at baseband. (c) PSD of the TX and RX signals. (d) RX signal at baseband (after cropping). (e) Normalized template transforms before and after residual \ac{TO} correction. (f) Zeros of the \ac{JBMOCZ} polynomial following residual \ac{TO} correction ($\numZeros=127$).}	
		\label{fig:sdr_results}
	\end{figure*}
	
	\subsection{SDR Implementation} \label{subsec:implementation}
	
	In this section, we transmit and receive $\numPayloadBits=424$~bits using \ac{OFDM}-\ac{BMOCZ} and off-the-shelf, ADALM-PLUTO \acp{SDR}, which are equipped with an AD9363 RF front end~\cite{adalmpluto}.\footnote{A version of this demo is available on GitHub, along with other example codes related to this work: https://github.com/parkerkh/Jutted-BMOCZ.} Note that the payload size $\numPayloadBits$ is too large to transmit in one shot (i.e., using a single polynomial) with \ac{BMOCZ}. Hence, we synthesize a multi-polynomial packet with \ac{FM} following the hybrid structure of Section~\ref{subsec:hybrid}, where we~set $\numBits=106$, $\numZeros=127$, $\Kprime=\lfloor\numZeros/2\rfloor=63$, and we encode each block of $\numBits$~bits using a (127,106)-BCH code. The payload thus comprises $\numSubcarriers=\numZeros+1=128$ active subcarriers and $\numOFDMsymbols=\lceil\numPayloadBits/\numBits\rceil=4$ \ac{OFDM} symbols. We use the radii $\radius_\text{sync}=1.025$, $\radius_\mathrm{J}=1.018$, and $\radius_\mathrm{H}=1.012$ for the synchronization preamble, \ac{JBMOCZ} symbol ($\asymFactor=1.03$), and Huffman \ac{BMOCZ} payload, respectively. Since $(\numZeros,\radius_\mathrm{J},\asymFactor)$ does not satisfy~\eqref{eq:sufficient_condition}, we obtain the template \ac{PAPR} numerically as $7.27$~dB, while the \ac{PAPR} of the preamble and payload are computed via~\eqref{eq:huffman_papr} as $1.50$~dB and $1.48$~dB, respectively. The sampling rate is fixed to $\sampleRate=20$~MHz, and we set $\idftSize=512$ and $\cpSize=8$ so that the symbol and \ac{CP} times~are $\symbolDuration=25.6$~$\mu$s and $\cpDuration=0.4$~$\mu$s. The packet duration is $\packetDuration=(\symbolDuration+\cpDuration)(\numOFDMsymbols+1)=0.13$~ms, and the bandwidth is $W=(\numSubcarriers+1)/\symbolDuration\approx5$~MHz. We utilize a carrier at $910$~MHz in the ISM band, and the transmit and receive gains set to $-5$~dB and $15$~dB, respectively. In our implementation, the receiver knows the basic system parameters (e.g., $\numBits$, $\numZeros$, $\idftSize$) but \emph{not} the \ac{CSI}.
	
	Fig.~\ref{fig:sdr_results} summarizes the results of our experiment, where we \ac{TX} and \ac{RX} the $424$-bit payload. Fig.~\ref{fig:sdr_results}(a) shows the \ac{SDR} setup, while Fig.~\ref{fig:sdr_results}(b) plots the \ac{I} and \ac{Q} components of the \ac{TX} signal. The synchronization preamble, \ac{JBMOCZ} timing symbol, and Huffman \ac{BMOCZ} payload have been annotated for reference. Fig.~\ref{fig:sdr_results}(c) shows the \ac{PSD} of both the \ac{TX} and \ac{RX} signals at baseband. The spikes in the \ac{PSD}, located at the edges of the occupied bandwidth, are due to the fact that the first and last coefficients for Huffman \ac{BMOCZ} (placed on the first and last subcarriers) carry nearly half of the sequence energy~\cite{walk2019principles,walk2020practical}. Fig.~\ref{fig:sdr_results}(d) plots the \ac{I} and \ac{Q} components of the \ac{RX} signal after coarse synchronization, \ac{CFO} compensation, and cropping. Fig.~\ref{fig:sdr_results}(e) shows the sampled templates in the time domain before and after correcting the residual \ac{TO}. Notice that the \ac{RX} template is shifted by $3$~samples, which represents the timing error following coarse synchronization with $\percentage=0.99$. Finally, Fig.~\ref{fig:sdr_results}(f) plots the \ac{TX} and \ac{RX} zeros of the \ac{JBMOCZ} polynomial. The \ac{RX} zeros overlap with the \ac{TX} zeros due to the correctly identified residual \ac{TO} and the high \ac{SNR}, which we estimate as $18.9$~dB. Indeed, with the \ac{dizet} decoder in~\eqref{eq:hd_dizet}, we recovered error-free both the header data of $\Kprime=63$~bits and the payload of $\numPayloadBits=424$~bits, achieving a data rate of $D=(\Kprime+\numPayloadBits)/\packetDuration=3.75$~Mbps and a throughput of $D/W=0.743$~bits/s/Hz.
	
	\section{Conclusion} \label{sec:conclusion}
	
	In this work, we propose \ac{JBMOCZ}, a generalization of Huffman \ac{BMOCZ} to include an asymmetry parameter. Fundamentally, our motivation for introducing asymmetry to the Huffman \ac{BMOCZ} zero constellation is to remove the ambiguity associated with zero rotation and thereby enable \ac{CPC}-free decoding. The proposed method, however, leads to a trade-off between asymmetry and zero stability. Accordingly, we introduce a reliability metric to measure zero stability and apply it to optimize the \ac{JBMOCZ} zero constellation parameters. Then, building on this analysis, we propose an \ac{OFDM}-\ac{BMOCZ} framework for the flexible transmission of multi-polynomial packets. We derive the \ac{PAPR} in closed form, and we describe a technique for blind channel estimation to improve the performance under frequency-selective fading.
	
	The contributions of this work lead to several natural research directions. For example, beyond its relevance to \ac{OFDM}-\ac{BMOCZ}, it is mathematically interesting to analyze how the roots of a polynomial behave when its coefficients are perturbed multiplicatively. Similarly, one can study the roots under $N$-point circular convolution of the coefficients, which is equivalent to pointwise multiplication after a \ac{DFT} and corresponds to polynomial multiplication modulo $z^N-1$. In addition to these considerations, the reliability metric of Section~\ref{subsec:stability} also has interesting applications. For instance, as Remark~\ref{remark:optimization} suggests, one can apply it to study the ``optimal" \ac{MOCZ} design with respect to zero stability, or to develop a reliability sequence for polar-coded \ac{MOCZ}. Finally, a practical research direction is to study extensions not considered in the system model of this work, such as Doppler, multiple users, and multiple antennas.
	
	\appendices
	
	\section{Proof of Proposition~\ref{prop:max_condition}} \label{app:prop_proof}
	
	\begin{proof}
		Suppose $\Ajutted(\e^{j\omega})$, given in~\eqref{eq:jutted_psd}, has a global maximum on the interval $[0,2\pi)$. We seek a sufficient condition such that $\argmax_{\omega\in[0,2\pi)}=0$. Since $\Ajutted(\e^{j\omega})$ is even with period $2\pi$, it suffices to consider the smaller interval $[0,\pi)\subset[0,2\pi)$. Let us ignore the scaling $(\numZeros+1)(\etaJ/\etaH)$ in~\eqref{eq:jutted_psd} and define $F(\omega) \triangleq E(\omega)S(\omega)$, where 
		\begin{equation}
			E(\omega) = \frac{2\cos(\omega)-a}{2\cos(\omega)-b} \ \ \text{and} \ \ S(\omega) = 1 - 2\etaH\cos(\numZeros\omega)\,,
		\end{equation}
		and $a,b\in\R$ are defined as in Proposition~\ref{prop:max_condition}. Now compute
		\begin{equation}
			E'(\omega) = -\frac{2(a-b)\sin(\omega)}{(2\cos(\omega)-b)^2} \ \ \text{and} \ \ S'(\omega) = 2\etaH\numZeros\sin(\numZeros\omega)\,.
		\end{equation}
		Since $\cos(\numZeros\omega)\geq-1$, we have $S(\omega)\leq1+2\etaH=S(\pi/\numZeros)$. Similarly, $E'(\omega)<0$ on $\omega\in(0,\pi)$, which implies that $E(\omega)\leq E(\pi/\numZeros)$ for all $\omega\in[\pi/\numZeros,\pi)$. Then, combining we obtain $F(\omega)\leq F(\pi/\numZeros)$ on $\omega\in[\pi/\numZeros,\pi)$, and hence $\argmax_{\omega\in[0,\pi)}F(\omega)\notin(\pi/\numZeros,\pi)$. Now note that because $F(\omega)$ is even, it necessarily has a local extremum at $\omega=0$. Hence, a sufficient condition for $\argmax_{\omega\in[0,\pi/\numZeros]}F(\omega)=0$ is $F'(\omega) = E'(\omega)S(\omega) + E(\omega)S'(\omega)<0$ on $\omega\in(0,\pi/\numZeros)$, which holds if and only if
		\begin{equation} \label{eq:fg_inequality}
			\underbrace{\frac{2\etaH\numZeros\sin(\numZeros\omega)}{1-2\etaH\cos(\numZeros\omega)}}_{\triangleq\,f(\omega)} < \underbrace{\frac{2(a-b)\sin(\omega)}{[a-2\cos(\omega)][b-2\cos(\omega)]}}_{\triangleq\,g(\omega)}\,.
		\end{equation}
		Recall that $2<b<a$. Thus, for $\omega\in(0,\pi/\numZeros)$, $\sin(\numZeros\omega)>0$ ensures $f(\omega)>0$. Similarly, for $\omega\in(0,\pi/\numZeros)$, $\sin(\omega)>0$ implies $g(\omega)>0$. We proceed by bounding $f(\omega)$ and $g(\omega)$ from above and below, respectively.
		
		\emph{Upper bound on $f$:} Since $\cos(\numZeros\omega)<1$ on $\omega\in(0,\pi/\numZeros)$, it follows that $1-2\etaH\cos(\numZeros\omega)>1-2\etaH$. Moreover, $\sin(\numZeros\omega)<\numZeros\sin(\omega)$ for all $\omega\in(0,\pi/\numZeros)$, which yields 
		\begin{equation}
			f(\omega) < \frac{2\etaH\numZeros^2\sin(\omega)}{1-2\etaH}\,.
		\end{equation}
		
		\emph{Lower bound on $g$:} Since $a-2\cos(\omega) < a-2\cos(\pi/\numZeros)$ and $b-2\cos(\omega) < b-2\cos(\pi/\numZeros)$ on $\omega\in(0,\pi/\numZeros)$, we have
		\begin{equation}
			g(\omega) > \frac{2(a-b)\sin(\omega)}{[a-2\cos(\pi/\numZeros)][b-2\cos(\pi/\numZeros)]}\,.
		\end{equation}
		Therefore,~\eqref{eq:fg_inequality} holds for all $\omega\in(0,\pi/\numZeros)$ if 
		\begin{equation}
			\frac{2\etaH\numZeros^2\sin(\omega)}{1-2\etaH} < \frac{2(a-b)\sin(\omega)}{[a-2\cos(\pi/\numZeros)][b-2\cos(\pi/\numZeros)]}\,,
		\end{equation}
		and rearranging we obtain the inequality in~\eqref{eq:sufficient_condition}.
	\end{proof}
	
	\bibliographystyle{IEEEtran} 
	\bibliography{references}
	
\end{document}